\documentclass[10pt, conference, compsocconf]{IEEEtran}
\IEEEoverridecommandlockouts
\ifCLASSINFOpdf
\else
\fi
\usepackage{multicol}
\usepackage[cmex10]{amsmath}
\usepackage{tablefootnote}
\usepackage[T1]{fontenc} 
\usepackage[utf8]{inputenc}
\usepackage{lmodern} 
\normalfont 
\DeclareFontShape{T1}{lmr}{bx}{sc} { <-> ssub * cmr/bx/sc }{}
\usepackage{dsfont} 
\usepackage{booktabs}
\usepackage{multirow}
\usepackage{graphicx}
\usepackage{chato-notes}
\usepackage[numbers,sort&compress,square]{natbib}
\usepackage[hidelinks]{hyperref}
\usepackage{siunitx}
\usepackage{algorithm}
\floatname{algorithm}{ALGORITHM}
\usepackage[noend]{algpseudocode}
\algrenewcommand\alglinenumber[1]{\small #1:}
\usepackage{pifont}
\usepackage{subfigure} 
\usepackage{xspace}

\usepackage{amsmath}
\usepackage{amsthm}
\usepackage{mathtools}
\usepackage{amssymb}
\usepackage{wasysym}

\usepackage{tikz}
\usetikzlibrary{arrows}
\definecolor{zzttqq}{rgb}{0.6,0.2,0}
\definecolor{qqqqff}{rgb}{0,0,1}
\definecolor{xdxdff}{rgb}{0.49019607843137253,0.49019607843137253,1}
\definecolor{cqcqcq}{rgb}{0.7529411764705882,0.7529411764705882,0.7529411764705882}

\newtheorem{theorem}{Theorem}[section]

\newtheorem{lemma}[theorem]{Lemma}

\usepackage{xifthen}

\newcommand{\algosingle}{\textsc{FourEst}}
\newcommand{\algomulti}{$\textsc{TS4C}_1$}
\newcommand{\algomultione}{$\textsc{TS4C}_1$}
\newcommand{\algomultitwo}{$\textsc{TS4C}_2$}

\newcommand{\algodynamic}{\textsc{ATS4C}}
\newcommand{\algomultifive}{\textsc{TS5C}}
\newcommand{\algosinglefive}{\textsc{FiveEst}}

\newcommand{\Samp}{\mathcal{S}}

\newcommand{\Prob}[1]{\textrm{Pr}\left(#1\right)}

\newcommand{\Exp}[1]{\text{E}\left[#1\right]}

\newcommand{\Var}[1]{\text{Var}\left[#1\right]}
\newcommand{\Cov}[1]{\text{Cov}\left[#1\right]}

\newcommand{\clisett}[2]{\mathcal{C}_{#1}^{(#2)}}

\newcommand{\fcliest}[1][]{%
\ifthenelse{\isempty{#1}{}}{\varkappa}{\varkappa^{(#1)}}%
}

\newcommand{\ntris}[1][]{%
\ifthenelse{\isempty{#1}{}}{\tau}{\tau^{(#1)}}%
}

\newcommand{\minone}[1]{\min\{1,#1\}}

\newcommand*{\neigh}{\mathcal{N}}

\newcommand{\tiesa}{\textsc{TieredSampling}}

\graphicspath{ {images/} }
\begin{document}
%
\title{
Tiered  Sampling: An Efficient Method for Approximate Counting \\ Sparse Motifs in Massive Graph Streams}


\author{\IEEEauthorblockN{Lorenzo De Stefani}
\IEEEauthorblockA{Department of Computer Science\\
Brown University\\
Providence, RI, USA\\
lorenzo@cs.brown.edu}
\and
\IEEEauthorblockN{Erisa Terolli$^*$}
\IEEEauthorblockA{Department of Computer Science\\
Sapienza University of Rome\\
Rome, Italy\\
terolli@di.uniroma1.it}
\thanks{*Work done in part while visiting Brown University.}
\and
\IEEEauthorblockN{Eli Upfal}
\IEEEauthorblockA{Department of Computer Science\\
Brown University\\
Providence, RI, USA\\
eli@cs.brown.edu}
}

\maketitle

\begin{abstract}
We introduce {\sc Tiered Sampling}, a novel technique for approximate counting sparse motifs in massive graphs whose edges are observed in a stream. Our technique requires only a single pass on the data and uses a memory of fixed size $M$, which can be magnitudes smaller than the number of edges.  

Our methods addresses the challenging task of counting sparse motifs - sub-graph patterns that have low probability to appear in a sample of $M$ edges in the graph, which is the maximum amount of data available to the algorithms in each step. To obtain an unbiased and low variance estimate of the count we partition the available memory to tiers (layers) of reservoir samples. While the base layer is a standard reservoir sample of edges, other layers are reservoir samples of sub-structures of the desired motif. By storing more frequent sub-structures of the motif, we increase the probability of detecting an occurrence of the sparse motif we are counting, thus decreasing the variance and error of the estimate.  

We demonstrate the advantage of our method in the specific applications of counting sparse 4 and 5-cliques in massive graphs.
%
%
We present a complete analytical analysis and 
extensive experimental results using both synthetic and real-world data. Our results demonstrate the advantage of our method in obtaining high-quality approximations for the number of 4 and 5-cliques for large graphs using a very limited amount of memory, significantly outperforming the single edge sample approach for counting sparse motifs in large scale graphs.

\end{abstract}


\IEEEpeerreviewmaketitle

\section{Introduction}
Counting motifs (sub-graphs with a given pattern)
in large graphs is a fundamental primitive in graph mining with numerous practical applications including link prediction and recommendation, community detection~\cite{BerryHLVP11}, topic
mining~\cite{EckmannM02}, spam and anomaly detection~\citep{BecchettiBCG10,lim2015} and protein interaction networks analysis~\cite{milo2002}.

Computing exact count of motifs in massive, Web-scale networks is often
impractical or even infeasible. Furthermore, 
many interesting networks, such as social networks, are {continuously growing},
 hence there is a limited value in maintaining an exact count.
The goal is rather to have, \emph{at any time}, a high-quality
approximation of the quantity of interest. 

To obtain a scalable and efficient solution for massive size graphs we focus here on the well studied model of one-pass stream computing. Our algorithms use a memory of fixed size $M$, where $M$ is significantly smaller that the size of the input graph. The input is given as a \emph{stream} of edges in an \emph{arbitrary} order, and the algorithm has only one pass on the input. The goal of the algorithm is to compute at any given time an unbiased, low variance estimate of the count of motif occurrences in the graph seen up to that time.
 
Given its theoretical and practical importance, the problem of counting motifs in graph streams has received great attention in the literature, with particular emphasis on the approximation of the number of 3-cliques (triangles)~\cite{destefani2017triest,pavan2013,KutzkovP14}. 
A standard approach to this problem is to sample up to $M$ edges uniformly at random, using a fixed sampling probability or, more efficiently,  reservoir sampling. A count of the number of motifs in the sample, extrapolated (normalized) appropriately,  gives an unbiased estimate for the number of occurrences in the entire graph. The variance (and error) of this method depends on the expected number of occurrences in the sample. In particular, for sparse motifs that are unlikely to appear many times in the sample, this method exhibits high variance (larger than the actual count) which makes it useless for counting sparse motifs. Note that when the input graph is significantly larger than the memory size $M$, a motif that is unlikely to appear in a random sample of $M$ edges may still have a large count in the graph. Also, as we attempt to count larger structure than triangles, these structures are more likely to be sparse in the graph. It is therefore important to obtain efficient methods for counting \emph{sparser} motifs in massive scale graph streams.

In this work we introduce the concept of {\sc Tiered Sampling} in stream computing. To obtain an unbiased and low variance estimate for the amount of sparse motif in massive scale graph we partition the available memory to tiers (layers) of reservoir samples. The base tier is a standard reservoir sample of individual edges, other tiers are reservoir samples of sub-structures of the desired motif.  This strategy significantly improves the probability of detecting occurrences of the motif. 

Assume that we count motifs with $k$ edges. If all the available  memory is used to store a sample of the edges, we would need $k-1$ of the motif's edges to be in the sample when the last edge of the motif is observed on the stream. The probability of this event decreases exponentially in $k$. 

Assume now that
 we use part of the available memory to store a sample of the observed occurrences of a fixed sub-motifs with $k/2$ edges. We are more likely to observe such motifs (we only need $k/2-1$ of there edges to be in the edge sample when the last edge is observed in the stream), and they are more likely to stay in the second reservoir sample since they are still relatively sparse. We now observe a full motif when the current edge in the stream completes an occurrence of the motif with edges sub-motifs in the two reservoirs.
  For an appropriate choice of parameters, and with fixed total memory size, this event has significantly higher probability than observing the $k-1$ edges in the edge sample.
To obtain an unbiased estimate of the count, the number of observed occurrences needs to be carefully normalized by the probabilities of observing each of the components. 

In this work we make the following main contributions:
\begin{itemize}
\item We introduce the concept of {\sl Tiered Sampling} for counting sparse concepts in large scale graph stream using multi-layer reservoir samples.
\item We develop and fully analyze two algorithms for counting the number of 4-cliques in a graph using 
two tier reservoir sampling.
\item For comparison purpose we analyze a standard (one tier) reservoir sample algorithm for 4-cliques counting problem. \item
We verify the advantage of our multi-layer algorithms by analytically comparing their performance to a standard (one tier) reservoir sample algorithm for 4-cliques counting problem on random Bar\'abasi-Albert graphs.
%
     \item We develop the \algodynamic{} technique, which allows to adaptively adjust the sub-division of memory space among the two tiers according to the properties of the graph being considered.
     \item We conduct an extensive experimental evaluation of the 4-clique algorithms on massive 
    graphs with up to hundreds of millions edges. We show the quality of the achieved estimations by comparing them with the actual ground truth value.  Our algorithms are also
    extremely scalable, showing update times in the order of hundreds  of microseconds for graphs with billions of edges.
     \item We demonstrate the generality of our approach through a second application of the two-tier method, estimating the number of 5-cliques in a graph stream using a second reservoir sample of the 4-cliques observed on the stream.
\end{itemize}
To the best of our knowledge these are the first fully analyzed, one pass stream algorithms for the 4 and 5-cliques counting problem.

\section{Preliminaries}
\label{sec:preliminaries}

For any (discrete) time step $t\ge 0$, we denote the graph observed up to and including time $t$ as $G^{(t)}=(V^{(t)},E^{(t)})$, where $V^{(t)}$ (resp., $E^{(t)}$) denotes the set of vertices (resp., edges) of $G^{(t)}$.
At time $t=0$ we have $V^{(t)}=E^{(t)}=\emptyset$. For any $t>0$, at time $t+1$ we receive one single edge $e_{t+1}=(u,v)$  from a stream, where $u,v$ are two distinct vertices. $G^{(t+1)}$ is thus obtained by \emph{inserting the new edge}: $E^{(t+1)}=E^{(t)}\cup\{(u,v)\}$; if either $u$ or $v$ do not belong to $V^{(t)}$, they are added to $V^{(t+1)}$. 
Edges can be added just once (i.e., we do not consider \emph{multigraphs} in this work) in an arbitrary adversarial order, i.e., as to cause the worst outcome for the algorithm. We however assume that the adversary has \emph{no access to the random bits} used by the algorithm.

This work explores the idea of storing a sample of sub-motifs in order to enhance the count of a sparse motif. For concreteness we focus on estimating the counts of 4-cliques and 5-cliques. 

Given a graph $G^{(t)}=(V^{(t)},E^{(t)})$, a \emph{$k$-clique} in $G^{(t)}$ is a \emph{set} of ${k\choose 2}$ (distinct) edges connecting a set of $k$ (distinct) vertices.
%
We denote by $\clisett{k}{t}$ the set of \emph{all} $k$-cliques in $G^{(t)}$. 

Our work makes use of the \emph{reservoir sampling scheme}~\citep{Vitter85}.
Consider a stream of elements $e_i$ observed in discretized time steps. Given a fixed sample  size $M>0$, for any time step $t$ the reservoir sampling scheme allows to maintain a uniform sample $\Samp$ of size $\min \{M,t\}$ of the $t$ elements observed on the stream:
\begin{itemize}
  \item If $t\le M$, then the element $e_t=(u,v)$ on the stream at time $t$ is
    deterministically inserted in $\Samp$.
  \item If $t>M$, then the sampling mechanism flips a biased coin with heads probability $M/t$. If
    the outcome is heads, it chooses an element $e_i$ uniformly at
    random from those currently in $\Samp$ which is replaced by $e_t$.
    Otherwise, $\Samp$ is not modified.
\end{itemize}


When using reservoir sampling for estimating the sub-graph count it is necessary to compute the probability of multiple elements being in $\Samp$ at the same time.

\begin{lemma}[Lemma 4.1~\cite{destefani2017triest}]\label{lem:reservoirhighorder}
  For any time step $t$ and any positive integer $k\le t$, let $B$ be any
  subset of size $|B|=k\le \min\{M,t\}$ of the element observed on the stream. Then, \emph{at the end} of time step $t$ (i.e., after updating the sample at time $t$), we have $\Pr(B\subseteq\Samp)=1$ if $t\le M$, and  $\Pr(B\subseteq\Samp)= \prod_{i=0}^{k-1}\frac{M-i}{t-i}$ otherwise.
  
  
\end{lemma}

In our experimental analysis (Sections~\ref{sec:expran} and~\ref{sec:experiments}) we measure the accuracy of the obtained estimator \emph{through the evolution of the graph} in terms of their \emph{Mean Average Percentage Error} (MAPE)~\cite{hyndman2006}. The MAPE measures the relative error of an estimator (in this case, $\varkappa^{(t)}$) with respect to the ground truth (in this case, $|\clisett{k}{t}|$) averaged over $t$ time steps: that is $MAPE = \frac{1}{t}\sum_{i=1}^t \frac{|\fcliest[t] - |\clisett{k}{t}||}{|\clisett{k}{t}|}$.

\section{Related Work}
Counting subgraphs in large networks is a  well studied problem in data mining which was originally brought to attention in the seminal work of Milo et al.~\cite{milo2002}. In particular, many contributions in the literature have focused on the triangle counting problem, including exact algorithms, MapReduce algorithms \cite{pagh2012, park2013} and streaming algorithms \cite{destefani2017triest, ahmed2014,  jha2015,pavan2013}. 

Previous works in literature on counting graph motifs \cite{rahman2014, ahmed2015} can also be used to estimate the number of cliques in large graphs. Other recent works on clique counting introduced randomized \cite{jain2017} and MapReduce \cite{finochi2015} algorithms. These require however priori information on the graph such as its degeneracy (for~\cite{jain2017}) or the vertex degree ordering (for~\cite{finochi2015}). Further, none of these approaches can be used in the streaming setting. 

The idea of using sub-structures of a graph motif in order to improve the estimation of its frequency in a massive graph has been previously explored in literature. In~\cite{bordino2008} Bordino et al. proposed a data stream algorithm  which estimates the number of occurrences of a given subgraph by sampling its ``\emph{prototypes} (i.e., sub-structures). While this approach is shown to be effective in estimating the counts of motif with three and four edges, it requires multiple passes through the graph stream and further knowledge on the properties of the graph. In \cite{jha2015cliques}, Jha et al. proposed an algorithm which effectively and efficiently approximates the frequencies of all 4-vertex subgraphs by sampling paths of length three. This algorithm requires however knowledge of the degrees of all the vertices in the graph and cannot be used in the streaming setting. 

To the best of our knowledge, our work is the first that proposes a  sampling-based, one pass algorithm for insertion only streams to approximate the global number of cliques found in large graphs. Further, our algorithms do not require any further information on the properties of the graph being observed. 

Using a strategy similar to our \tiesa{} approach, in~\cite{jha2015} Jha and Seshadri propose a one pass streaming algorithm for triangle counting which using a first reservoir for edges which are then used  to generate a stream of \emph{wedges} (i.e., paths of length two) stored in a second reservoir. This approach appear to be not worthwhile for triangle counting as it is consistently outperformed by a simpler strategy based on a single reservoir presented~\cite{destefani2017triest}. This is due to the fact that as in most large graph of interest wedges are much more frequent than edges themselves it is not worth devoting a large fraction of the available memory space to maintaining wedges over edges.

\section{\tiesa{} application to 4-clique counting}
\label{sec:multi}
\begin{figure}
\begin{center}
	\begin{tikzpicture}[line cap=round,line join=round,>=triangle 45,x=1cm,y=1cm, scale=1, every node/.style={scale=1}]
	\fill[line width=1.2pt,color=zzttqq,fill=zzttqq,fill opacity=0.10000000149011612] (0,3) -- (0,0) -- (3,0) -- cycle;
	
	\fill[line width=0pt,color=qqqqff,fill=qqqqff,fill opacity=0.1] (3,3) -- (0,0) -- (3,0) -- cycle;
	
	\draw [line width=1.2pt] (0,3)-- (0,0);
	
	\draw [line width=1.2pt] (0,0)-- (3,0);
	
	\draw [line width=1.2pt] (3,0)-- (3,3);
	
	\draw [line width=1.2pt,dash pattern=on 1pt off 3pt] (3,3)-- (0,3);
	
	\draw [line width=1.2pt] (3,3)-- (0,0);
	
	\draw [line width=1.2pt] (0,3)-- (3,0);
	
	\begin{scriptsize}\draw [fill=xdxdff] (0,3) circle (2pt);
	
	\draw[color=xdxdff] (-0.19,3.16) node {$u$};
	
	\draw [fill=qqqqff] (3,3) circle (2pt);
	
	\draw[color=qqqqff] (3.16,3.16) node {$v$};
	\draw [fill=qqqqff] (0,0) circle (2pt);
	\draw[color=qqqqff] (-0.16,-0.19) node {$z$};
	\draw [fill=xdxdff] (3,0) circle (2pt);
	\draw[color=xdxdff] (3.16,-0.19) node {$w$};
	\draw[color=black] (-0.2,1.5) node {$e_{2}$};
	\draw[color=black] (1.5,-0.2) node {$e_{1}$};
	\draw[color=black] (3.2,1.5) node {$e_{3}$};
	\draw[color=black] (1.5,3.2) node {$e_{6}$};
	\draw[color=black] (0.95,2.3) node {$e_{4}$};
	\draw[color=black] (2.05,2.3) node {$e_{5}$};
	\draw[color=zzttqq] (0.68,1.5) node {$T_1$};
	\draw[color=qqqqff] (2.44,1.5) node {$T_2$};
	\end{scriptsize}\end{tikzpicture}
	\caption{Detection of 4-clique using triangles}
	\label{fig:4clidetection}
	\end{center}
	\vskip -.3in
\end{figure}
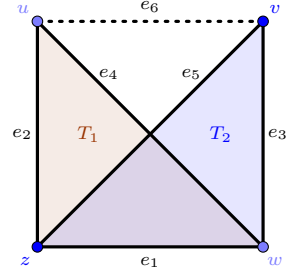

In this section we present \algomulti{} and \algomultitwo{}, two applications of our \tiesa{} approach for  counting the number of 4-cliques in an undirected graph observed as an edge stream. The two algorithms partition the available memory into two samples, an edge reservoir sample and a triangle reservoir sample. \algomulti{} attempts in each step to construct a 4-cliques using the current observed edge, two edges from the edge reservoir sample and one triangle from the triangle reservoir sample. \algomultitwo{} attempts at each step to construct a 4-clique from the current observed edge and two triangles from the triangle reservoir sample.
At each time step $t$, both algorithms maintain a \emph{running estimation} $\fcliest[t]$ of $|\clisett{4}{t}|$. Clearly $\fcliest[0]= 0$ and the estimator is increased every time a 4-clique is ``\emph{observed}'' on the stream. The two algorithms also maintains a counter $\ntris[t]$ for the number of triangles observed in the stream up to time $t$. This value is used by the reservoir sampling scheme which manages the triangle reservoir.  

\subsection{Algorithm \algomulti{} description}
\begin{algorithm}[t]
 \tiny 
  \caption{\algomulti{} - Tiered Sampling for 4-Clique counting}
  \label{alg:ncliquecounter}
  \begin{algorithmic}[0]
  \Statex{\textbf{Input:} Insertion-only edge stream $\Sigma$, integers $M$, $M_\Delta$}
    \State $\Samp_e \leftarrow\emptyset$ , $\Samp_\Delta \leftarrow\emptyset$, $t\leftarrow 0$, $t_\Delta\leftarrow 0$, $\sigma\leftarrow 0$
    \For{ {\bf each} element $(u,v)$ from $\Sigma$}
    \State $t\leftarrow t +1$
    \State \textsc{Update4Cliques}$(u,v)$
    \State \textsc{UpdateTriangles}$(u,v)$
    \If{\textsc{SampleEdge}$((u,v), t )$}
      \State $\Samp \leftarrow \Samp\cup \{((u,v), t)\}$
    \EndIf
    \EndFor
    \Statex
    \Function{Update4Cliques}{$(u,v), t$}
      \For{{\bf each} triangle $(u, w, z)\in \Samp_\Delta$}
        \If{$(v,w) \in \Samp_e \wedge (v,w) \in \Samp_e$}
      		\State $p \leftarrow \textsc{ProbClique}((u,w,z),(v,w), (v,z)) $
      		\State $\sigma \leftarrow \sigma + p^{-1}/2$
     \EndIf   
      \EndFor
      \For{{\bf each} triangle $(v, w, z)\in \Samp_\Delta$}
        \If{$(u,w) \in \Samp_e \wedge (u,z) \in \Samp_e$}
      		\State $p \leftarrow \textsc{ProbClique}((v,w,z),(u,w), (u,z)) $
      		\State $\sigma \leftarrow \sigma + p^{-1}/2$
     \EndIf   
      \EndFor
    \EndFunction
    \Statex
    \Function{UpdateTriangles}{$(u,v), t $}
      \State $\mathcal{N}^\Samp_{u,v} \leftarrow \mathcal{N}^\Samp_u \cap \mathcal{N}^\Samp_v$ 
      \For{ {\bf each} element $w$ from $\mathcal{N}^\Samp_{u,v}$}
       	\State $t_\Delta \leftarrow t_\Delta + 1$
      	\If{\textsc{SampleTriangle}$(u,v,w)$}
      		\State $\Samp_\Delta \leftarrow \Samp_\Delta\cup \{u,v,w\}$
      	\EndIf
      \EndFor
    \EndFunction
    \Statex
    \Function{\textsc{SampleTriangle}}{$u,v,w$} 
         \If {$t_\Delta\leq  M_\Delta$}
        \State \textbf{return} True
      \ElsIf{\textsc{FlipBiasedCoin}$(\frac{M_\Delta}{t_\Delta}) = $ heads}
        \State $(u_1,v_1,w_1) \leftarrow$ random triangle from $\Samp_\Delta$
        \State $\Samp_\Delta\leftarrow \Samp_\Delta\setminus \{(u_1,v_1,w_1)\}$
        \State \textbf{return} True
      \EndIf
      \State \textbf{return} False
    \EndFunction
    \Statex
    \Function{\textsc{SampleEdge}}{$(u,v),t$} 
      \If {$t\leq  M$}
        \State \textbf{return} True
      \ElsIf{\textsc{FlipBiasedCoin}$(\frac{M}{t}) = $ heads}
        \State $((u',v'),t') \leftarrow$ random edge from $\Samp$
        \State $\Samp\leftarrow \Samp\setminus \{((u',v'),t')\}$
        \State \textbf{return} True
      \EndIf
      \State \textbf{return} False
    \EndFunction
  \end{algorithmic}
\end{algorithm}

Algorithm \algomulti{} maintains an \emph{edges (resp., triangles) reservoir sample} $\Samp_e$ (resp., $\Samp_\Delta$) of \emph{fixed} size $M_e$ (resp., $M_\Delta$). From Lemma~\ref{lem:reservoirhighorder}, for any $t$ the probability of any edge $e$ (resp., triangle $T$) observed on the stream $\Sigma$ (resp., observed by the \algomulti{}) to be included in $S_e$ (resp., $S_\Delta$) is $M_e/t$ (resp. $M_\Delta/\tau^{(t)}$). We denote as $\Samp_e^{(t)}$ (resp., $\Samp_\Delta^{(t)}$) the set of edges (resp., triangles) in $\Samp_e$ (resp., $\Samp_\Delta$)  \emph{before} any update to the sample(s) occurring at step $t$. We denote with $\neigh^{\Samp_e^{(t)}}_u$ the \emph{neighborhood} of $u$ with respect to the edges in $\Samp_e^{(t)}$, that is $\neigh^{\Samp_e^{(t)}}_u =\{v\in V^{\Samp_e^{(t)}} ~:~ (u,v)\in \Samp_e^{(t)}\}$.

Let $e_t = (u,v)$ be the edge observed on the stream at time $t$. At each step \algomulti{} executes three main tasks:
\begin{itemize}
	\item \emph{Estimation update:} \algomulti{} invokes the function \textsc{Update 4-cliques} to detect any 4-clique completed by $(u,v)$  (see Figure~\ref{fig:4clidetection} for an example). That is, the algorithm verifies whether in triangle reservoir $\Samp_\Delta$ there exists any triangle $T$ which includes $u$ (resp., $v$). Note that since each edge is observed just once and the edge $(u,v)$ is being observed for the first time no triangle in $\Samp_\Delta$ can include both $u$ and $v$. For any such triangle $T'=\{u,w,z\}$ (or $\{v,w,z\}$) the algorithm checks whether the edges $(v,w)$ and $(v,z)$ (resp., $(v,w)$ and $(v,z)$) are currently in $\Samp_e$. When such conditions are meet, we say that a 4-clique is ``\emph{observed on the stream}''. The algorithm then uses \textsc{ProbClique} to compute the \emph{exact} probability $p$ of the observation based on the timestamps of all its edges. The estimator $\fcliest$ is then increased by $p^{-1}/2$.
	\item \emph{Triangle sample update:} using \textsc{UpdateTriangles} the algorithm verifies whether the edge $e_t$ completes any triangle with the edges in $\Samp_e^{(t)}$. If that is the case, we say that a new triangle $T^*$ is \emph{observed on the stream}. The counter $\tau$ is increased by one and the new triangle is a candidate for inclusion in $\Samp_\Delta$ with probability $M_\Delta/\tau^{t}$.
	\item \emph{Edge sample update:} the algorithm updates the edge sample $\Samp_e$ according to the Reservoir Sampling scheme described in Section~\ref{sec:preliminaries}.  
\end{itemize} 

Each time a 4-clique is observed on the stream, \algomulti{} uses \textsc{ProbClique} to compute the \emph{exact probability} of the observation. Such computation is all but trivial as it is influenced by both the order according to which the edges of the 4-clique were observed on the stream, and by the number of triangles observed on the stream $\ntris[t]$. The analysis proceeds using a (somehow tedious) analysis of \emph{all} the 5! possible orderings of the first five edges of the clique observed on the stream. Due to space limitations we refer the reader to Lemma~\ref{lem:probclique} of the extend version of this work~\cite{destefani2017} for the detailed computation.  

%
\subsection{Analysis of the estimator}
We now present the analysis of the estimations obtained using \algomulti{}. First we show their \emph{unbiasedness} and we then provide a bound on their variance. We refer the reader to Appendix~\ref{app:probability computation} of the extended version of this work for the complete proofs~\cite{destefani2017}. In the following we denote as $t_\Delta$  the first time step at for which the number of triangles seen by \algomultione{} exceeds $M_\Delta$.

\begin{theorem}\label{thm:multiunbiased}
$\varkappa^{(t)}=|\clisett{k}{t}|$ if $t \leq \min \{M_e, t_\Delta\}$, and
  $\mathbb{E}\left[\varkappa^{(t)}\right]$ $=|\clisett{k}{t}|$ otherwise.
\end{theorem}

We now show an \emph{upper bound} to the variance of the \algomulti{}
estimations. The proof relies on a very careful analysis of the order of arrival of the edge shared between a pair of two 4-cliques. We present here the most general result for $t\geq{M_e,t_\Delta}$. Note that if $t \leq \min \{M_e, t_\Delta\}$, from Theorem~\ref{thm:multiunbiased}, $\varkappa^{(t)}=|\clisett{k}{t}|$ and hence $\textrm{Var}\left[\varkappa^{(t)}\right] = 0$. For $\min \{M_e, t_\Delta\}<t\leq \max \{M_e, t_\Delta\}$, $\textrm{Var}\left[\varkappa^{(t)}\right]$ admits an upper bound similar to the one in Theorem~\ref{thm:variancemulti}.

\begin{theorem}\label{thm:variancemulti}
For any time $t>\min\{M_e,t_\Delta$\}, we have
  \begin{align*}
    \textrm{Var}\left[\varkappa^{(t)}\right]
    &\le
    |\clisett{4}{t}|\left(c\left(\frac{t-1}{M_e}\right)^4\left(\frac{\tau^{(t)}}{M_\Delta}\right)-1\right) \\
    &+2a^{(t)}\left(c\frac{t-1}{M_e}-1\right)\\ &+2b^{(t)}\left(c\left(\frac{t-1}{M_e}\right)^2\left(\frac{1}{4}\frac{\tau^{(t)}}{M_\Delta} + \frac{3}{4}\frac{t-1}{M_e} \right)-1\right),
  \end{align*}
  where $a^{(t)}$ (resp., $b^{(t)}$) denotes the number of unordered pairs of 4-cliques which share one edge (resp., three edges) in $G^{(t)}$, and $c\geq \frac{M_e^3}{(M_e-1)(M_e-2)(M_e-3)}$.
\end{theorem} 

\subsection{Memory partition across layers}\label{sec:memorypartition}
In most practical scenario we assume that a certain amount of total available memory $M$ is available for algorithm \algomulti{}. A natural question concern what is the best way of spitting the available memory between $\Samp_e$ and $\Samp_\Delta$. While different heuristics are possible, in our work we chose to assign the available space in such a way that the dominant first term of upper bound of \algomulti{} in Theorem~\ref{thm:variancemulti} is minimized. For $0<\alpha<1$ let $M_e = \alpha M$ and $M_\Delta = \left(1-\alpha\right)M$. Then we have that $|\clisett{4}{t}|\left(c\left(\frac{t-1}{\alpha M}\right)^4\left(\frac{\tau^{(t)}-1}{(1-\alpha)M}\right)-1\right)$ is minimized for $\alpha = 4/5$. This convenient splitting rule works well in most cases, and is used in most of the paper. Section~\ref{sec:adaptive} we discuss a more sophisticated dynamic allocation of memory, \algodynamic{}, and present experimental results for this method in Section~\ref{sec:exada}.
\subsection{Concentration bound}\label{sec:concent}
We now show a concentration result on the estimation of \algomultione{}, which relies on Chebyshev's inequality~\citep[Thm.~3.6]{mitzenmacher2005probability}.

\begin{theorem}\label{thm:multioneconcentration}
Let $t>\min\{M_e,t_\Delta$\} and assume $|\clisett{4}{t}|>0$. For any
  $\varepsilon,\delta\in(0,1)$, if
  \begin{align*}
    M> \max\Bigl\{&\frac{5}{4}\sqrt[5]{\frac{12c(t-1)^4(\ntris[t])}{\delta
      \varepsilon^2|\clisett{4}{t}|}},\frac{15c a^{(t)}(t-1)}{\delta
    \varepsilon^2 |\clisett{4}{t}|^2},\\
   &\frac{5}{4}\sqrt[3]{\frac{3b^{(t)}c\left(t-1\right)^2\left(4\ntris[t]+3(t-1)\right)}{2\delta
    \varepsilon^2 |\clisett{4}{t}|^2}} \Bigr\}
  \end{align*}
  then $|\varkappa^{(t)}-|\clisett{4}{t}|<\varepsilon|\clisett{4}{t}|$ with
  probability $>1-\delta$.
\end{theorem}

\subsection{Algorithm \algomultitwo{}}
\label{sec:algotwo}
\algomultitwo{} detects
a 4-clique when the current observed edge completes a 4-clique with \emph{two triangles} in $\Samp_\Delta$. 
\algomultione{} and \algomultitwo{} differ  \emph{only} in the function \textsc{Update4Cliques}, the pseudocode for \algomultitwo{} is presented in Algorithm~\ref{alg:ncliquecounter2}.

\begin{algorithm}[t]
  \tiny
  \caption{\algomultitwo{} - Tiered Sampling for 4-Clique counting using 2 triangle sub-structures}
  \label{alg:ncliquecounter2}
  \begin{algorithmic}[0]
      \Function{Update4Cliques}{$(u,v), t$}
      \For{{\bf each} $(u,w,z) \in \Samp_\Delta\wedge (v,w,z)\in\Samp_\Delta$}
      \State $p \leftarrow \textsc{ProbClique}((u,w,z),(v,w,z)) $
      		\State $\sigma \leftarrow \sigma + p^{-1}/2$  
      \EndFor
    \EndFunction
   \end{algorithmic}
\end{algorithm}

The probability of detecting a 4-clique using \algomultitwo{} computed by \textsc{ProbClique} is  different from the correspondent probability computed in \algomultione{}. For the details we refer the reader to Lemma~\ref{lem:probclique2} in Appendix~\ref{app:probability computation2} of the extended version~\cite{destefani2017}.
\begin{lemma}\label{lem: probalgo2}
	Let $\lambda\in \clisett{4}{t}$ and let $p_\lambda$ denote the probability of $\lambda$ being observed on the stream by \algomultitwo{}. We have:
\begin{equation*}
	p_\lambda \leq \left(\frac{M_e}{t-1}\right)^4\left(\frac{M_\Delta}{\tau^{(t)}}\right)^2.
\end{equation*}
\end{lemma}

Applying this lemma to an analysis similar to the one used in the proof of Theorem~\ref{thm:multiunbiased} we prove that the estimations obtained using \algomultitwo{} are \emph{unbiased}.
\begin{theorem}\label{thm:multiunbiased2}
 The estimator $\varkappa$ returned by \algomultitwo{} is \emph{unbiased}, that is:
  $\mathbb{E}\left[\varkappa^{(t)}\right]$ $=|\clisett{k}{t}|$.
\end{theorem}

The analysis presents several complications due to the interplay of the probabilities of observing each of the two triangles that share an edge. For the details of these results we refer the reader to Appendix~\ref{app:probability computation2} of the extended version of this work~\cite{destefani2017}.
Following the same criterion discussed in Section~\ref{sec:memorypartition}, we use $|\Samp_e|=2M/3$ and $|\Samp_\Delta| = M/3$ as a general rule for assigning the available memory space between the two sample levels.
Although the difference between \algomultione{} and \algomultitwo{} may appear of minor interest, our experimental analysis show that it can lead to significantly different performances depending on the properties of the graph $G^{(t)}$. Intuitively, \algomultitwo{} emphasizes the importance of the triangle sub-structures compared to \algomultione{}, thus resulting in better performance when the input graph is very sparse with smaller frequencies of 3 and 4-cliques.

\section{Comparison with single sample approach}\label{sec:sing}
To quantify the advantage of our \tiesa{} approach we construct and fully analyze algorithm \algosingle{} that uses a single sample strategy.
We then compare the performance of \algosingle{} and \algomultione{} analytically, and through experiments on both synthetic and real-world data.

\subsection{\algosingle{}}
 \algosingle{} is an extension of the reservoir sample triangle counting algorithm in \cite{destefani2017triest}, using one reservoir sample to store uniform random sample $\Samp$ of size $M$ of the edges observed over the stream. The pseudocode for \algosingle{} is presented in Algorithm~\ref{alg:ncliquecountersingle}.
 
%
%
%
\begin{algorithm}[ht]
\tiny
\caption{\algosingle{}}
  \label{alg:ncliquecountersingle}
   \begin{algorithmic}[0]
	\Statex{\textbf{Input:} Edge stream $\Sigma$, integer $M\ge6$}
	\Statex{\textbf{Output:} Estimation of the number of 4-cliques $\varkappa$}
    \State $\Samp_e \leftarrow\emptyset$, $t\leftarrow 0$, $\varkappa\leftarrow 0$ 
    \For{ {\bf each} element $(u,v)$ from $\Sigma$}
    \State $t\leftarrow t +1$
    \State \textsc{Update4Cliques}$(u,v)$
    \If{\textsc{SampleEdge}$((u,v), t )$} 
      \State $\Samp \leftarrow \Samp\cup \{(u,v)\}$
    \EndIf
    \EndFor
    \Statex
    \Function{Update4Cliques}{$u,v$}
      \State $\mathcal{N}^\Samp_{u,v} \leftarrow \mathcal{N}^\Samp_u \cap \mathcal{N}^\Samp_v$
      \For{ {\bf each} element $(x,w)$ from $\mathcal{N}^\Samp_{u,v} \times \mathcal{N}^\Samp_{u,v}$}
       \If{$(x,w)$ in  $\Samp_e$} 
       		\If{$t\leq M$}
       			\State{$p\leftarrow 1$}
       		\Else
      			\State $p \leftarrow  min \{1, \frac{M(M-1)(M-2)(M-3)(M-4)}{(t-1)(t-2)(t-3)(t-4)(t-5)}\}$
      		\EndIf
      	\State $\varkappa \leftarrow \varkappa + p^{-1}$
    \EndIf
	\EndFor
   \EndFunction
\end{algorithmic}
\end{algorithm}

The proofs of the following results can be found in the extended version of this work~~\cite{destefani2017}.
\begin{theorem}\label{thm:singleunbiased}
  Let $\varkappa^{(t)}$ the estimated number of 4-cliques in $G^{(t)}$ computed by \algosingle{} using memory of size $M$. $\varkappa^{(t)}=|\clisett{4}{t}|$ if $t\le M+1$ and
  $\mathbb{E}\left[\varkappa^{(t)}\right]$ $=|\clisett{4}{t}|$ if $t> M+1$.
\end{theorem}

We now show an \emph{upper bound} to the variance of the \algosingle{}
estimations for $t>M$ (for $t\leq M$ we have $\varkappa^{(t)}=|\clisett{4}{t}|$ and thus the variance of $\varkappa^{(t)}$ is zero).

\begin{theorem}\label{thm:variancesingle}
  For any time $t>M+1$, we have
  \begin{align*}
    \textrm{Var}\left[\varkappa^{(t)}\right]
    &\le
    |\clisett{4}{t}|\left(\left(\frac{t-1}{M}\right)^5-1\right) \\
    &+a^{(t)}\left(\frac{t-1}{M}-1\right)+b^{(t)}\left(\left(\frac{t-1}{M_e}\right)^3-1\right).
  \end{align*}
\end{theorem} 

\subsection{Variance comparison}\label{sec:varcompa}
Although the upper bounds obtained in Theorems~\ref{thm:variancemulti} and~\ref{thm:variancesingle} cannot be compared directly, they still provide some useful insight on which algorithm may be performing better according to the properties of $G^{(t)}$. 

Let us consider the  first, dominant, terms of each of the variance bounds, that is $|\clisett{4}{t}\left(\left(\frac{5}{4}\frac{t}{M}\right)^4\frac{5\tau^{(t)}}{M}-1\right)|$ for \algomultione{} (assigning $M_e=4M/5$ and $M_\Delta = M/5$),  and $|\clisett{4}{t}|\left(\left(\frac{t-1}{M}\right)^5-1\right)$ for \algosingle{}.  
%
While \algomultione{} exhibits a slightly higher constant multiplicative term a cost due to the splitting of the memory in the \tiesa{} approach, the most relevant difference is however given by the term $\frac{\tau^{(t)}}{M}$ appearing in the bound for \algomultione{} compared with an additional $\frac{t-1}{M}$ appearing in the bound fro \algosingle{}. Recall that $\tau^{(t)}$ denotes here the number of triangles \emph{observed} by the algorithm up to time $t$. Due to the fact that the probability of observing a triangle decreases quadratically with respect to the size of the graph $t$, we expect that $\tau<t$ and, for sparser graphs for which 3 and 4-cliques are indeed ``\emph{rare patterns}, we actually expect $\tau^{(t)}<<t$. Under these circumstances we would therefore expect $M/5\tau >> M/t$.

 This is the critical condition for the success of the \tiesa{} approach. If the sub-structure selected as a tool for counting the motif of interest is not ``\emph{rare enough}'' then there is no benefit in devoting a certain amount of the memory budget to maintaining a sample of occurrences of the sub-structure. Such problem would for instance arise when using the \tiesa{} approach for counting triangles using \emph{wedges} (i.e., two-hop paths) as a sub-structure, as attempted in~\cite{jha2015}, as in most real-world graph the number of wedges is much greater of the number of edges themselves making them not suited to be used as a sub-structure.

\subsection{Experimental evaluation over random graphs}\label{sec:expran}

In this section we compare the performances of our \tiesa{} algorithms of \algomultione{} and \algomultitwo{} with the performance of the a single sample approach \algosingle{}, on randomly generated graphs. In particular, we analyze random Bar\'abasi-Albert graphs~\cite{albert2002} with $n$ nodes and $mn$ total edges, as they exhibit the same \emph{scale-free} property observed in many real-world graphs of interest such as social networks. 

In our experiments, we set $n=20000$ and we consider various values for $m$ from 50 to 2000 in order to compare the performances of the two approaches as the number of edges (and thus triangles) increases. The algorithms use a memory space whose size corresponds to 5\% of the number of edges in the graph $nm$. In columns 2,3 and 5 of Table~\ref{tab:mapereduction} we present the average of the MAPEs of the ten runs for \algosingle{}, \algomultione{} and \algomultitwo{}. In column 4 (resp., 6) of Table~\ref{tab:mapereduction} we report the percent reduction/increase in terms of the average MAPE obtained by \algomultione{} (resp., \algomultitwo{}) with respect to \algosingle{}.  

%
Both \tiesa{} algorithms consistently outperform \algosingle{} for values of $m$ up to 400, that is for fairly sparse graphs for which we expect 3 and 4 cliques to be rare patterns. The advantage over \algosingle{} is particularly strong for values of $m$ up to 200 with reductions of the average MAPE up to 30\%. For denser graphs, i.e. $m\geq 750$, \algosingle{} outperforms \emph{both} \tiesa{} algorithms. This in consistent with the intuition discussed in Section~\ref{sec:varcompa}, as for denser graphs triangles are not ``\emph{rare enough}'' to be worth saving over edges. Note however that in these cases the quality of \emph{all} the estimators is very high (i.e., MAPE$\leq$ 1\%). 

We can also observe that \algomultitwo{} outperforms \algomultione{} for $m\leq 200$. Vice versa, \algomultione{} outperforms \algomultitwo{} (and \textsc{FourEst}) for $300\leq m\leq 750$. Since, as discussed in Section~\ref{sec:algotwo}, \algomultitwo{} gives ``\emph{more importance}'' to triangles, it works particularly well when the graph is very sparse, and triangles are particularly rare.  As the graph grows denser (and the number of triangles increases), \algomultione{} performs better until, for highly dense graphs \algosingle{} produces the best estimates.

\begin{table}[]
\centering
\tiny
\renewcommand{\arraystretch}{1.3}
\begin{tabular}{l|c|cc|cc}
\hline
m    & \algosingle{} & \algomultione{} & \begin{tabular}{@{}c@{}}Change \\ \algomultione{}\end{tabular} & \algomultitwo{}& \begin{tabular}{@{}c@{}}Change \\ \algomultitwo{}\end{tabular} \\
\hline
50   & 0.8775    & 0.5862   & -33.19\%   & 0.5222 & -40.49\%   \\
100  & 0.3054    & 0.1641   & -46.27\%   & 0.1408 & -53.90\%   \\
150  & 0.1521    & 0.0937   & -38.34\%   & 0.0917 & -39.70\%   \\
200  & 0.0899    & 0.0599   & -33.39\%   & 0.0549 & -38.95\%   \\
300  & 0.0486    & 0.0346   & -28.80\%   & 0.0417 & -14.19\%   \\
400  & 0.0289    & 0.0249   & -13.87\%   & 0.0261 & -9.32\%   \\
500  & 0.0221    & 0.0197   & -11.28\%   &0.0239 & 8.08\%    \\
750  & 0.0134    & 0.0132   & -1.57\%    & 0.0181 & 34.90\%    \\
1000 & 0.0088    & 0.0099   & 13.14\%    & 0.0146 & 66.36\%   \\
\hline
\end{tabular}
\caption{Comparison of MAPE of \algosingle{}. \algomultione{} and \algomultitwo{} for Bar\'abasi-Albert graphs.}\vskip -.3in 
\label{tab:mapereduction}
\end{table}

\section{Adaptive Tiered Sampling Algorithm}
\label{sec:adaptive}
An appropriate partition of the available memory between the layers used in the \tiesa{} approach is crucial for the success of the algorithm: while assigning more memory to the triangle sample allows maintain more sub-patterns, removing too much space from the edge sample reduces the probability of observing new triangles. 
While in Section~\ref{sec:concent} we provide a general rule to decide  how to split the memory budget  for \algomultione{} and \algomultitwo{}, such partition may not always lead to the best possible results. For instance, if the graph being observed is particularly sparse, assigning a large portion of the memory to the triangles would result in a considerable waste of memory space due to the low probability of observing triangles. Further, as discussed in Section~\ref{sec:sing}, depending on the properties of $G^{(t)}$ a single edge approach could perform better than the \tiesa{} algorithms. As in the graph streaming setting these properties are generally not known a priori, nor stable through the graph evolution, a fixed memory allocation policy appears not to be the ideal solution. 
In this section we present \algodynamic{}, an \emph{adaptive} variation of our \algomultitwo{} algorithm, which \emph{dynamically} analyzes the properties of $G^{(t)}$ through time and consequently decides how to allocate the available memory. 
\paragraph{Algorithm description}\algodynamic{} has two main ``\emph{execution regimens}'': \textbf{(R1)} for which it behaves exactly as \algosingle{}, and \textbf{(R2)} for which it behaves similarly to \algomultitwo{}. 
\textbf{(R1)} is the initial regimen for \algodynamic{}. Recall that, from Lemma~\ref{lem:reservoirhighorder} (resp., Lemma~\ref{lem: probalgo2}) the probability of a 4-clique being observed by  \algosingle{} (resp., \algomultitwo{})is upper bounded by   
$p_s=(\min\{1,M/t\})^5$ (resp., $p_\alpha=(\min\{1,\alpha M/t\})^4(\min\{1,(1-\alpha) M/\ntris[t]\})^2$, where $0<\alpha<1$). Every $M$ time steps the algorithm evaluates, based on the number of triangles observed so far, whether to switch to \textbf{(R2)}. \algodynamic{} $\alpha^* = \textrm{argmax}_{\alpha \in [2/3,1)} p_\alpha$. If $p_s>p_{\alpha^*}$, \algodynamic{} remains in \textbf{(R1)}. Vice versa \algodynamic{} transitions to \textbf{(R2)}: the triangle reservoir $\Samp_\Delta$ is assigned $(1-\alpha^*)M$ memory space, and it is filled with the triangles composed by the edges \emph{currently} in the edge reservoir, using the reservoir sampling scheme. Finally the edge reservoir $\Samp_e$ is constructed by selecting $\alpha^* M$ of the edges in the current sample uniformly at random, thus ensuring that $\Samp_e$ is an \emph{uniform sample}. Once \algodynamic{} switches to \textbf{(R2)} it never goes back to \textbf{(R1)}.
During (R2), as long as $|\Samp_\Delta|< M/3$, every $M$ time steps \algodynamic{} evaluates whether it is opportune to assign more memory to $\Samp_\Delta$. Let $t= iM$, rather than just using the information of the number of triangles seen so far $\ntris[iM]$, \algodynamic{} computes a ``\emph{prediction}'' of the total number of triangles seen until $(i+1)M$ assuming that the number of triangles seen during the next $M$ steps will equals the number of triangles seen during the last $M$ steps, that is $\tilde{\ntris[(i+1)M]} = 2\ntris[iM]-\ntris[(i-)M]$. \algodynamic{} then computes $\alpha^* = \textrm{argmax}_{a\in [2/3, 1)}(\min\{1,\alpha M/(i+1)M)\})^4(\min\{1,(1-\alpha) M/\tilde{\ntris[(i+1)M]}\})^2$. Let $\alpha$ denote the split being used by \algodynamic{} at $t=iM$, if $\alpha^*>\alpha$ the algorithm continues its execution with no further operations (\algodynamic{} never reduces the memory space assigned to $\Samp_\Delta$). Vice versa, if  $\alpha^*<\alpha$, \algodynamic{} removes $(\alpha - \alpha^*)M$ edges from $\Samp_e$ selected uniformly at random, and the freed space is assigned to $\Samp_\Delta$. Let us denote this space as $\Samp'_\Delta$. As \algodynamic{} progresses and observes new triangles it fills $\Samp'_\Delta$ using the reservoir sampling scheme. $\Samp_\Delta$ and $\Samp'_\Delta $ are then merged at the first time step for which the probability $p_\Delta$ of a triangle seen before the creation of $\Samp'_\Delta$ being in $\Samp_\Delta$ becomes lower then the probability $p_{\Delta'}$ of a triangle observed after the creation of $\Samp'_\Delta$ being in it. The merged triangle sample contains all the triangles in $\Samp'_\Delta$, while the triangles in $\Samp_\Delta$ are moved to it with probability $p_{\Delta'}/p_{\Delta}$. This ensures that after the merge all the triangles seen on the stream are kept in the triangle reservoir with probability $p_{\Delta'}$. After the merge, \algodynamic{} operates the samples as described in \algomultitwo{}. Finally, \algodynamic{} increases the memory space for $\Samp_\Delta$ only if all the currently assigned space is used. 

The analysis of \algodynamic{} is much more complicated than the one for our previous algorithms as it requires to keep track of the probabilities according to which the observed triangles appear in $\Samp_\Delta$. We claim however (without explicit proof) that the estimation provided by \algodynamic{} is unbiased. 
Via the experimental analysis in Section~\ref{sec:exada}, we show how \algodynamic{} succeeds in merging the advantages of the single sample approach and of an adaptive splitting of the in tiers.

\section{Experimental Evaluation}
\label{sec:experiments}
In this section, we evaluate through extensive experiments the performance of our proposed \tiesa{} method when applied for counting 4 and 5-cliques in large graphs observed as streams. We use several real-world graphs with size ranging form $10^6$ to $10^8$ edges (see Table~\ref{tab:graphs} for a complete list).  All graphs are treated as undirected. The edges are observed on the stream according to the values of the associated timestamps if available, or in random order otherwise.  In order to evaluate the accuracy of our algorithms, we compute the ``\emph{ground truth}'' exact number of 4-cliques (resp., 5-cliques) for each time step using an exact algorithm which maintains the entire $G^{(t)}$ in memory.
Our algorithms are implemented in Python. The experiments were run on the Brown University CS department cluster\footnote{https://cs.brown.edu/about/system/services/hpc/grid/}, where each run employed a single core and used at most 4GB of RAM.
\begin{table}[t]
\tiny
\renewcommand{\arraystretch}{1.3}
\centering
\begin{tabular}{l c c c}
\hline
Graph & Nodes & Edges & Source\\
\hline
DBLP & 986,324 & 3,353,618 & \cite{boldi2011} \\
Patent (Cit) &  2,745,762 & 13,965,132 & \cite{destefani2017triest}\\
LastFM &  681,387 & 30,311,117 & \cite{destefani2017triest}\\
Live Journal & 5,363,186 & 49,514,271 & \cite{boldi2011}\\
Hollywood & 1,917,070 & 114,281,101 & \cite{boldi2011}\\
Orkut & 3,072,441& 117,185,083 & \cite{mislove2007} \\
Twitter & 41,652,230 & $1.20 \cdot 10^9$ & \cite{destefani2017triest, boldi2011} \\
\hline
\end{tabular}
\caption{Graphs used in the experiments}
\vskip -.3in 
\label{tab:graphs}
\end{table}
The section is organized as follows: we first evaluate the performance of \algomultione{} and \algomultitwo{} when run on several massive graphs and we compare them with the estimations obtained using \algosingle{}. We then present practical examples that motivate the necessity for the adaptive version of our \tiesa{} approach, and we show how our \algodynamic{} manages to capture the best of the single and multi-level approach.
Finally, we show how the \tiesa{} approach can be generalized in order to count  structures other than 4-cliques.

\subsection{Counting 4-Cliques}
\label{sec:estimation}
We estimate the global number of 4-cliques on insertion-only streams, starting as empty graphs and for which an edge is added at each time step, using algorithms \algosingle{}, \algomultione{} and \algomultitwo{}. As discussed in Section~\ref{sec:memorypartition} (resp., Section~\ref{sec:algotwo}),in \algomultione{} (resp., \algomultitwo{})  we split the total available memory space $M$ as $|\Samp_e|=4M/5$ and $|\Samp_\Delta| = M/5$ (resp., $|\Samp_e|=2M/3$ and $|\Samp_\Delta| = M/3$). 

The experimental results show that these fixed memory splits perform well for most cases. We then experiment with an adaptive splitting mechanism that handles the remaining cases. 
In Figure~\ref{fig:evolution} we present the estimation obtained by averaging 10 runs of respectively \algomultione{}, \algomultitwo{} and \algosingle{} using total memory space $M=5\times 10^5$ for the LiveJournal and Hollywood graphs (i.e., respectively using  less than 1\% and 0.5\% of the graph size).  While the average of the runs for \algomultione{} and \algomultitwo{} are almost \emph{indistinguishable} from the ground truth, that is clearly not the case for \algosingle{} for which the quality of the estimator considerably worsens as the graph size increases.

In Table~\ref{tab:mape}, we report the average MAPE performance over 10 runs for \algomultione{}, \algomultitwo{} and \algosingle{} for several graphs of different size. Depending on the size of the input graph we assign different total memory space (as reported in Table~\ref{tab:mape}), which in most cases amounts to at most 3\% ($\sim8$\%M for  Patent(Cit)) of the input graph size. Except for LastFM, our \tiesa{} algorithms clearly outperform \algosingle{}, with the average MAPE reduced by up to 30\%. Both $\textsc{TS4C}_1$ and $\textsc{TS4C}_2$ return very accurate estimations of the number of 4-cliques on the majority of these graphs, with average MAPE lower than 10\% (16\% for LastFM and Orkut). 
	LastFM is the only graph for which \algosingle{} (considerably) outperforms the \tiesa{} algorithms. This is due to the high density of the graph $|E|/|V| > 500$ and to the fact that in this case triangles are not a rare enough sub-structure to justify the choice of maintaining them over simple edges. 
\begin{table}[]
\tiny
\centering
\renewcommand{\arraystretch}{1.3}
\begin{tabular}{l|c|c|c|c}
\hline
Dataset    & $M$ &  \algosingle{} & \algomultione{}  & \algomultitwo{} \\
\hline
Patent (Cit)   & $10^6$& 0.0963    & 0.0921      & 0.0474   \\
LastFM      &$10^6$ & 0.0118    & 0.1258     & 0.0777   \\
LiveJournal & $10^6$&  0.1521    & 0.0543     & 0.0560   \\
Hollywood   & $2\cdot10^6$& 0.0355    & 0.0207     & 0.0194   \\
Orkut       &$2\cdot10^6$& 0.4674    & 0.1590    & 0.1417   \\
Twitter\tablefootnote{Ground truth computed for first $3\cdot10^8$ edges on the stream.} &$5\cdot10^6$ &   0.2503 &	0.0749  & 0.0742  \\
\hline
\end{tabular}
\caption{Average MAPE of various approaches for all graphs.}
\vskip -.2in 
\label{tab:mape}
\end{table} 

\begin{figure}[t]
  \centering
   \subfigure[Live Journal]{\includegraphics[scale=0.295]{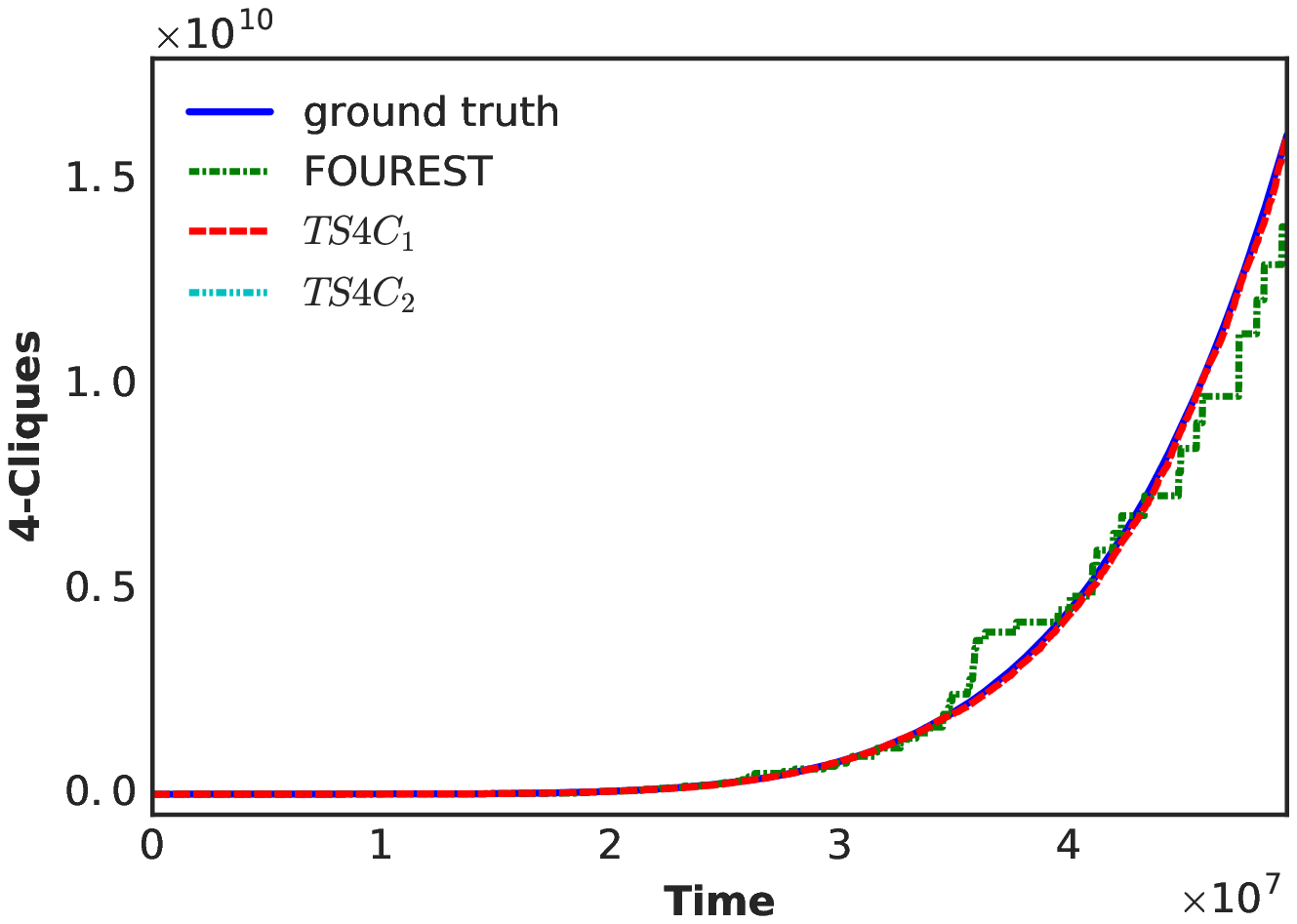}}\quad
   \subfigure[Hollywood]{\includegraphics[scale=0.295]{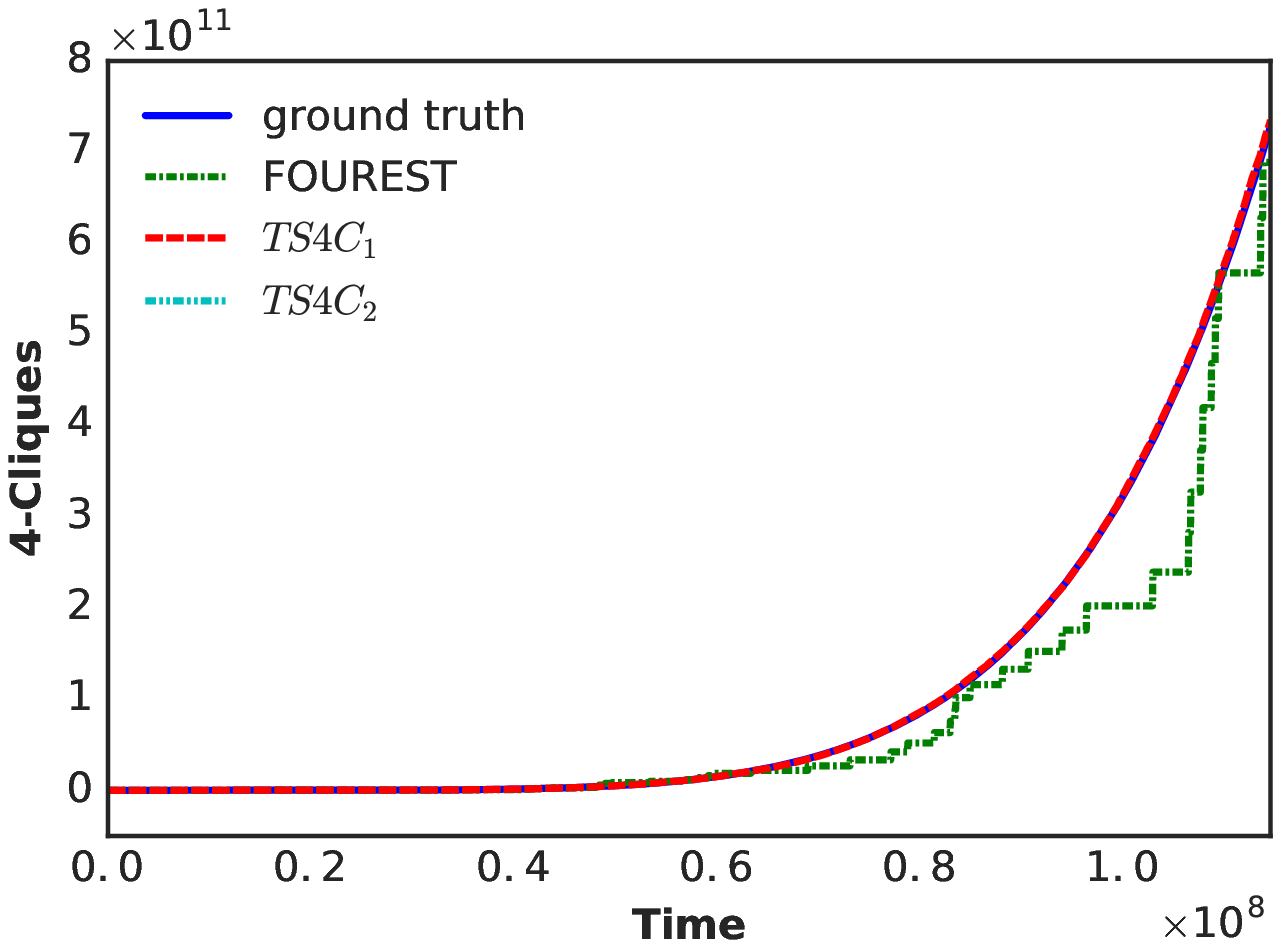}}\quad
   \caption{Comparison of $|\clisett{4}{t}|$ estimates obtained using \algomultione{}, \algomultitwo{} and \algosingle{} with $M=5 \times 10^5$,}
    \label{fig:evolution}
    \vskip -.2in 
\end{figure}
We analyze the variance reduction achieved using our \tiesa{} algorithms by comparing the empirical variance observed over forty runs on the Hollywood graph using $M= 5\times 10^5$. The results are reported in Figure~\ref{fig:variance}. While for both \algomultione{} and \algomultitwo{} the minimum and maximum estimators are close to the ground truth throughout the evolution of the graph, \algomultione{} estimators exhibit very high variance especially towards the end of the stream. 

Our experiments not only verify that \algomultione{} and \algomultitwo{} allow to obtain good quality estimations which are in most cases superior to the ones achievable using a single sample strategy, but also validate the general intuition underlying the \tiesa{} approach.

Both \tiesa{} algorithms are extremely scalable, showing average update times in the order of hundreds microseconds for all graphs.

\begin{figure*}[t]
  \centering
   \subfigure{\includegraphics[scale=0.4]{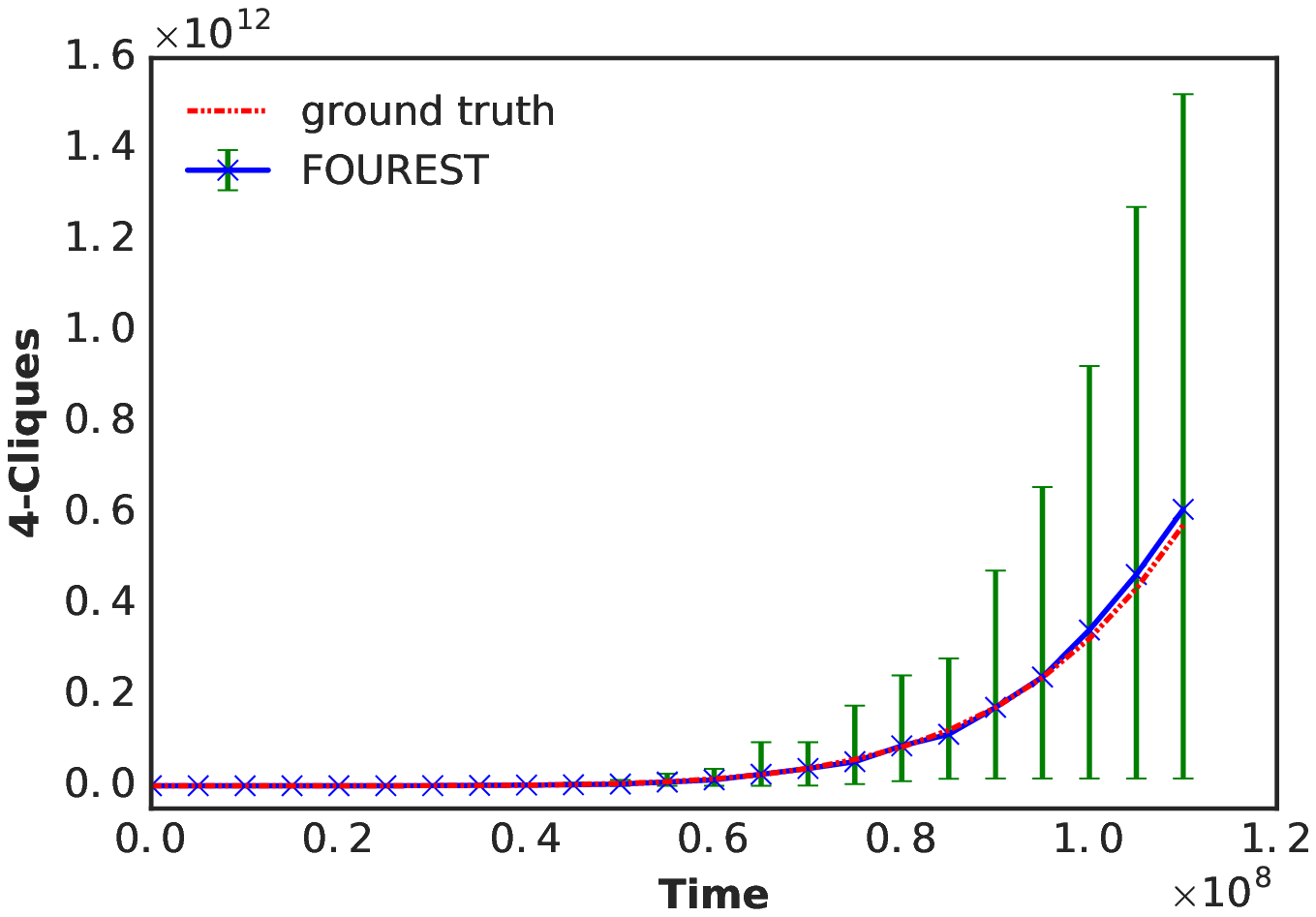}}\quad
    \subfigure{\includegraphics[scale=0.4]{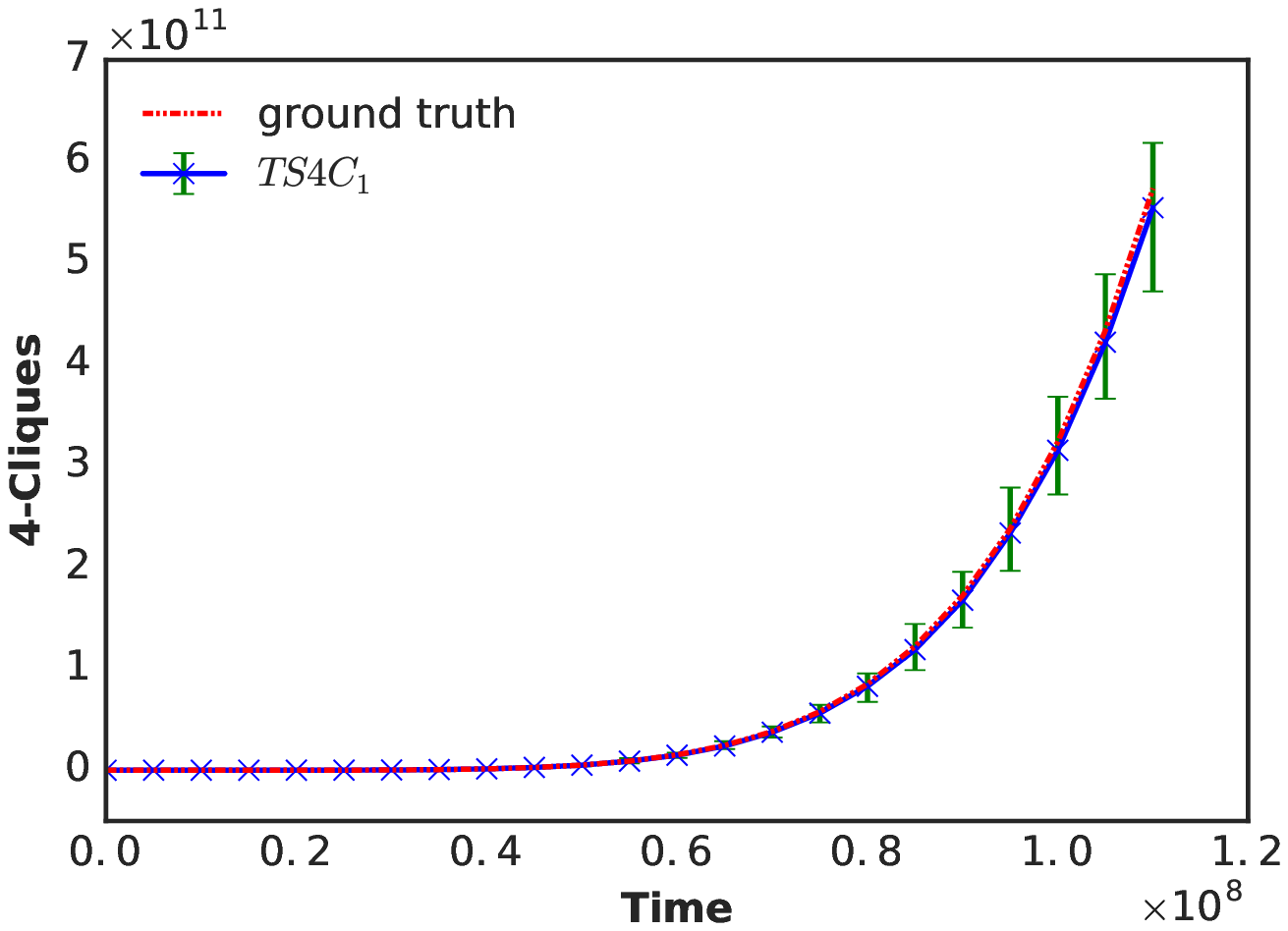}}\quad
     \subfigure{\includegraphics[scale=0.4]{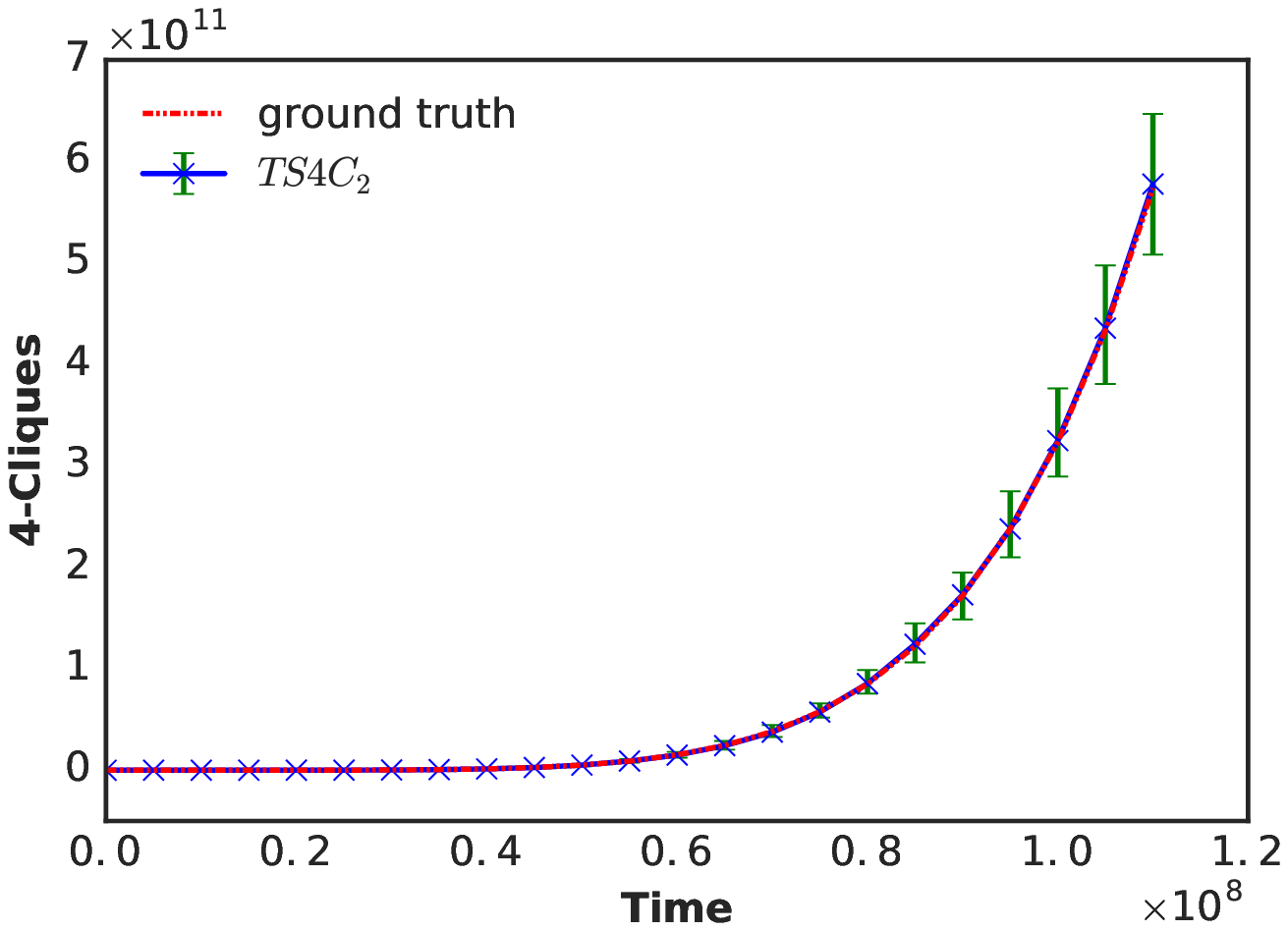}}\quad
    
   \caption{Variance of \algosingle{}, \algomultione{} and \algomultitwo{} for Hollywood graph with $M=5\times 10^5$.}
   \vskip -.1in 
    \label{fig:variance}
\end{figure*}

\subsection{Adaptive Tiered Sampling}\label{sec:exada}
In Section~\ref{sec:estimation}, we showed that \algomultitwo{}, allows to obtain high quality estimations for the number of 4-cliques outperforming in most cases both \algomultione{} and \algosingle{}. These results were obtained splitting the available memory such that $|\Samp_e|=2M/3$ and $|\Samp_\Delta| = M/3$. As discussed in Section~\ref{sec:adaptive}, while this is an useful general rule, depending on the properties of the graph different splitting rules may yield better results. We verify this fact by evaluating the performance of \algomultitwo{} when run on the Patent(Cit) graph using different assignments of the total space $M=5\times 10^5$ to the two levels. The results in Figure~\ref{fig:differentsplitting} show that decreasing the space assigned to the triangle sample from $M/3$ to $M/9$ allows to achieve estimations which are closer to the ground truth leading to a 31\% reduction in the average MAPE. 
Due to the sparsity of the Patent(Cit) graph, \algomultitwo{} observes a very small number of triangles for large part of the stream. Assigning a large fraction of the memory space to $\Samp_\Delta$ is thus inefficient as the probability of observing new triangles is reduced, and the space assigned to $\Samp_\Delta$ is not fully used.
\begin{figure}[t]
\centering
\includegraphics[scale=0.4]{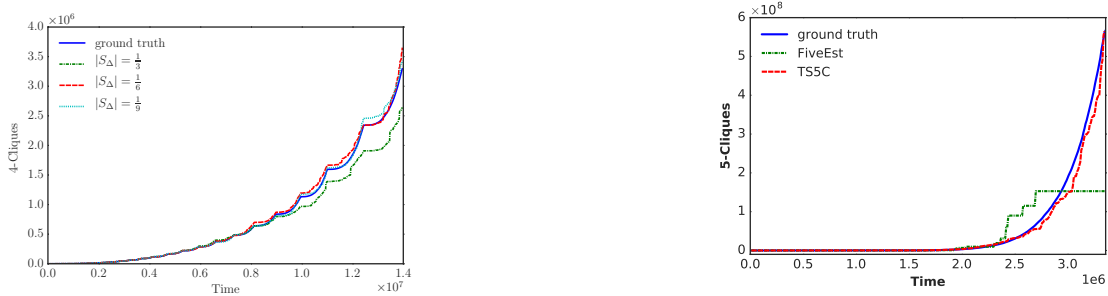}
\caption{$|\clisett{4}{t}|$ estimations for Patent (Cit)  using \algomultitwo{} with $M = 5\times 10^5$ and different memory space assignments among tiers.}
\vskip -.2in 
\label{fig:differentsplitting}
\end{figure}

To overcome such difficulties, in Section~\ref{sec:adaptive} we introduced \algodynamic{}, an \emph{adaptive} version of \algomultitwo{}, which allows to dynamically adjust the use of the available memory space based on the properties of the graph being  observed. We experimentally evaluate the performance of \algodynamic{} over 10 runs on the Patent(Cit) and the LastFM graphs and we compare it with \algomultitwo{} and \algosingle{} using $M=5\times 10^5$. 
\begin{figure}[t]
  \centering
   \subfigure[Patent(Cit)]{\includegraphics[scale=0.295]{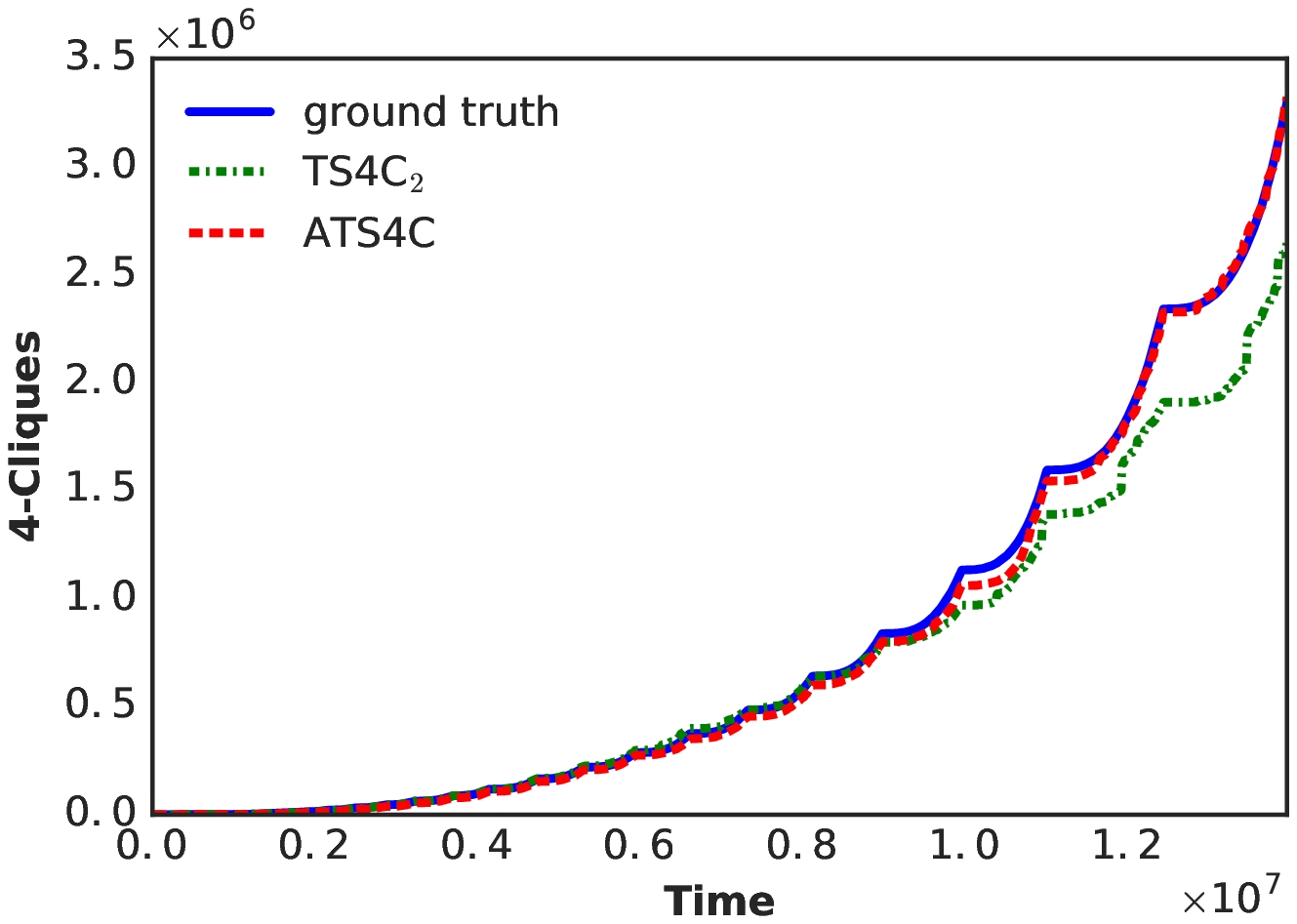}}\quad
   \subfigure[LastFM]{\includegraphics[scale=0.295]{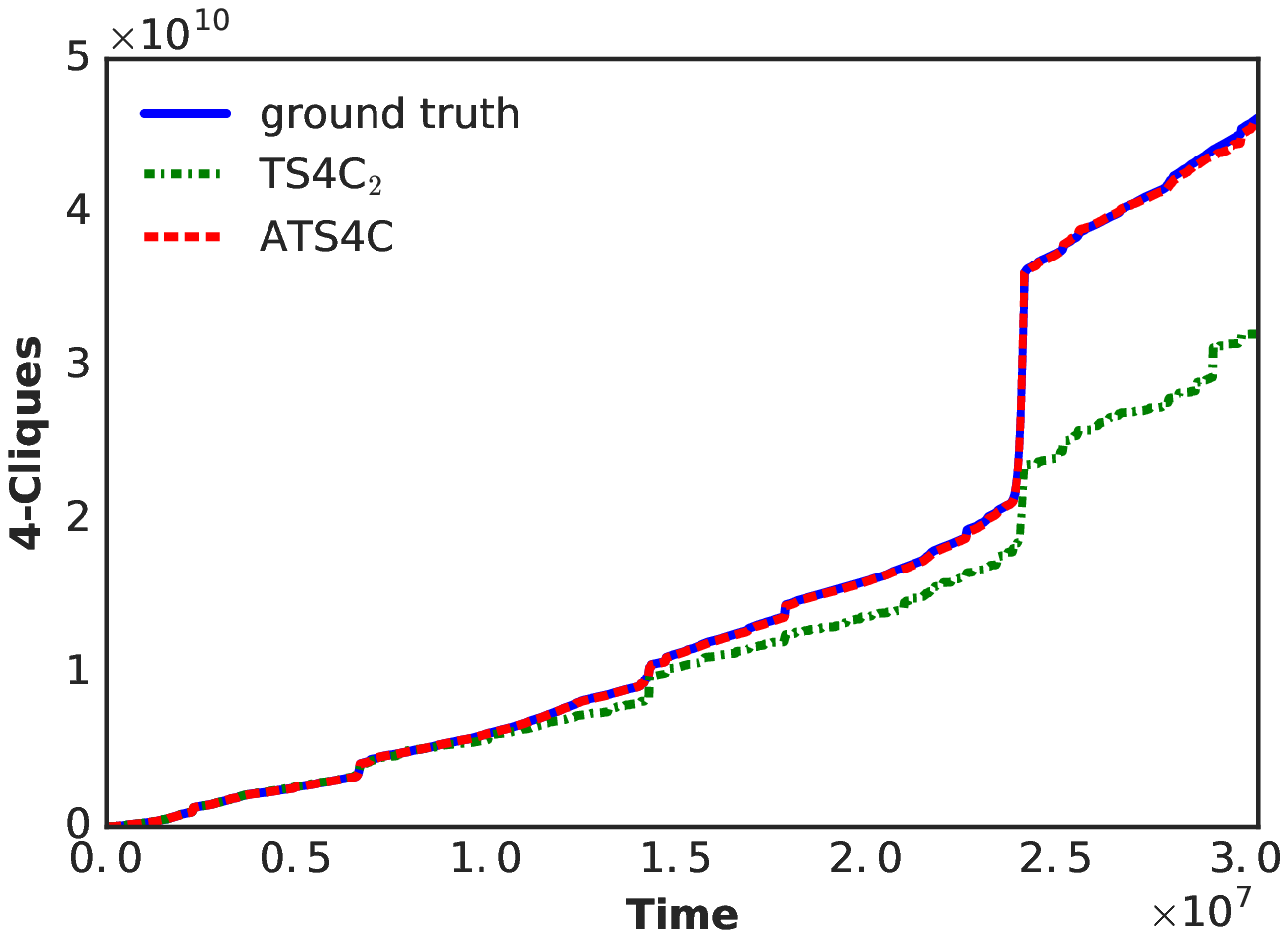}}
   \caption{Comparison of $|\clisett{4}{t}|$ estimates obtained using \algodynamic{}, \algomultitwo{} and \algosingle{} with $M=5 \times 10^5$.}
   \vskip -.224in 
    \label{fig:dynamic}
\end{figure}

As shown in Figure~\ref{fig:dynamic},  for both graphs, \algodynamic{} produces estimates that are nearly indistinguishable from the ground truth. \algodynamic{} clearly outperforms \algomultitwo{} (with $|\Samp_\Delta|=\frac{M}{3}$) on Patent(Cit) where triangles are sparse motifs achieving an $\sim85$\% reduction in the average MAPE compared to \algomultitwo{}. \algodynamic{} returns high quality estimations even for the LastFM graph, for which the single level approach \textsc{FourEst} outperforms the \tiesa{} algorithms.

\subsection{Counting 5-Cliques}

To demonstrate the generality of our \tiesa{} approach we present \algomultifive{}, a one pass counting algorithms for 5-clique in a stream. The algorithm maintains two reservoir samples, one for edges and one for 4-cliques. When the current observed edge completes a 4-clique with edges in the edge sample the algorithm attempts to insert it to the 4-cliques reservoir sample. A 5-clique is counted when the current observed edge completes a 5-clique using one 4-clique in the reservoir sample and 3 edges in the edge sample. 

%
Our experiments compare the performance of \algomultifive{} to that of a standard one tier edge reservoir sample algorithm \algosinglefive{}, similar to  \algosingle. We evaluate the average performance over 10 runs of the two algorithms on the DBLP graph using $M= 3\times 10^5$. The results are presented in Figure~\ref{fig:fivecliques}. \algomultifive{} clearly outperforms \algosinglefive{} in obtaining much better estimations of the ground truth value $|\clisett{5}{t}|$  achieving an $\sim 56$\% reduction in the average MAPE.
\begin{figure}[t]
\centering
\includegraphics[scale=0.4]{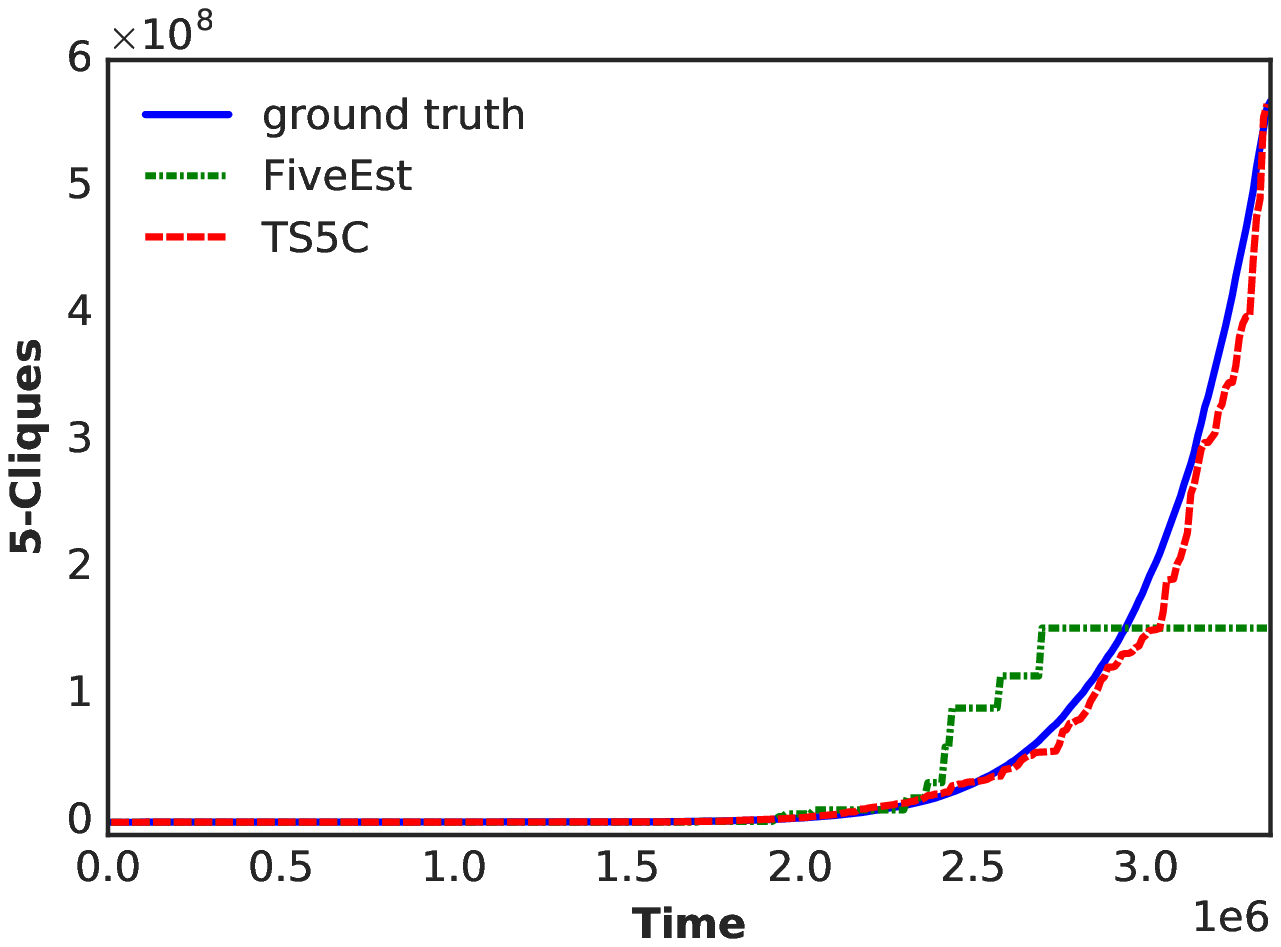}
\caption{Comparison of $|\clisett{5}{t}|$ estimates for the DBLP graph obtained using \algomultifive{} and \algosinglefive{} with $M=3\times10^5$.}
\vskip -.2in 
\label{fig:fivecliques}
\end{figure}

\section{Conclusions}
We developed \tiesa{}, a novel technique for approximate counting sparse motifs in massive graphs whose edges are observed in one pass stream. We studied application of this technique for the problems of counting 4 and 5-cliques. We present both analytical proofs and experimental results, demonstrating the advantage of our method in counting sparse motifs compared to the standard methods of using just edge reservoir sample. With the growing interest in discovering and analyzing large motifs in massive scale graphs in social networks, genomics, and neuroscience, we expect to see further applications of our technique.

\section*{Acknowledgement}
This work was partially supported by NSF grant IIS-1247581,  by a Google Focused Research Award, by the Sapienza Grant C26M15ALKP, by the SIR Grant RBSI14Q743, and by the ERC Starting Grant DMAP 680153.

\bibliographystyle{IEEEtran}

\bibliography{main}

\begin{thebibliography}{10}
\providecommand{\url}[1]{#1}
\csname url@samestyle\endcsname
\providecommand{\newblock}{\relax}
\providecommand{\bibinfo}[2]{#2}
\providecommand{\BIBentrySTDinterwordspacing}{\spaceskip=0pt\relax}
\providecommand{\BIBentryALTinterwordstretchfactor}{4}
\providecommand{\BIBentryALTinterwordspacing}{\spaceskip=\fontdimen2\font plus
\BIBentryALTinterwordstretchfactor\fontdimen3\font minus
  \fontdimen4\font\relax}
\providecommand{\BIBforeignlanguage}[2]{{%
\expandafter\ifx\csname l@#1\endcsname\relax
\typeout{** WARNING: IEEEtran.bst: No hyphenation pattern has been}%
\typeout{** loaded for the language `#1'. Using the pattern for}%
\typeout{** the default language instead.}%
\else
\language=\csname l@#1\endcsname
\fi
#2}}
\providecommand{\BIBdecl}{\relax}
\BIBdecl

\bibitem{BerryHLVP11}
J.~W. Berry, B.~Hendrickson, R.~A. LaViolette, and C.~A. Phillips, ``Tolerating
  the community detection resolution limit with edge weighting,''
  \emph{Physical Review E}, vol.~83, no.~5, p. 056119, 2011.

\bibitem{EckmannM02}
J.-P. Eckmann and E.~Moses, ``Curvature of co-links uncovers hidden thematic
  layers in the {W}orld {W}ide {W}eb,'' \emph{PNAS}, vol.~99, no.~9, pp.
  5825--5829, 2002.

\bibitem{BecchettiBCG10}
L.~Becchetti, P.~Boldi, C.~Castillo, and A.~Gionis, ``Efficient algorithms for
  large-scale local triangle counting,'' \emph{ACM TKDD}, pp. 13:1--13:28,
  2010.

\bibitem{lim2015}
Y.~Lim and U.~Kang, ``Mascot: Memory-efficient and accurate sampling for
  counting local triangles in graph streams,'' in \emph{KDD'15}.\hskip 1em plus
  0.5em minus 0.4em\relax ACM, 2015, pp. 685--694.

\bibitem{milo2002}
R.~Milo, S.~Shen-Orr, S.~Itzkovitz, N.~Kashtan, D.~Chklovskii, and U.~Alon,
  ``Network motifs: simple building blocks of complex networks,''
  \emph{Science}, vol. 298, no. 5594, pp. 824--827, 2002.

\bibitem{destefani2017triest}
L.~De~Stefani, A.~Epasto, M.~Riondato, and E.~Upfal, ``Tri{\`e}st: Counting
  local and global triangles in fully dynamic streams with fixed memory size,''
  \emph{ACM TKDD'17}, vol.~11, no.~4, p.~43, 2017.

\bibitem{pavan2013}
A.~Pavan, K.~Tangwongsan, S.~Tirthapura, and K.-L. Wu, ``Counting and sampling
  triangles from a graph stream,'' \emph{VLDB'13}, pp. 1870--1881, 2013.

\bibitem{KutzkovP14}
K.~Kutzkov and R.~Pagh, ``Triangle counting in dynamic graph streams,'' in
  \emph{SWAT '14}.\hskip 1em plus 0.5em minus 0.4em\relax Springer, 2014, pp.
  306--318.

\bibitem{Vitter85}
J.~S. Vitter, ``Random sampling with a reservoir,'' \emph{TOMS}, vol.~11,
  no.~1, pp. 37--57, 1985.

\bibitem{hyndman2006}
R.~J. Hyndman and A.~B. Koehler, ``Another look at measures of forecast
  accuracy,'' \emph{IJF}, vol.~22, no.~4, pp. 679--688, 2006.

\bibitem{pagh2012}
R.~Pagh and C.~E. Tsourakakis, ``Colorful triangle counting and a mapreduce
  implementation,'' \emph{Inf. Process. Lett.}, pp. 277--281, 2012.

\bibitem{park2013}
H.-M. Park and C.-W. Chung, ``An efficient mapreduce algorithm for counting
  triangles in a very large graph,'' in \emph{CIKM '13}, 2013, pp. 539--548.

\bibitem{ahmed2014}
N.~K. Ahmed, N.~Duffield, J.~Neville, and R.~Kompella, ``Graph sample and hold:
  A framework for big-graph analytics,'' in \emph{KDD '14}, 2014, pp.
  1446--1455.

\bibitem{jha2015}
M.~Jha, C.~Seshadhri, and A.~Pinar, ``A space-efficient streaming algorithm for
  estimating transitivity and triangle counts using the birthday paradox,''
  \emph{ACM TKDD}, vol.~9, no.~3, pp. 15:1--15:21, Feb. 2015.

\bibitem{rahman2014}
M.~Rahman, M.~A. Bhuiyan, and M.~Al~Hasan, ``Graft: An efficient graphlet
  counting method for large graph analysis,'' \emph{IEEE TKDE}, vol.~26,
  no.~10, pp. 2466--2478, 2014.

\bibitem{ahmed2015}
N.~K. Ahmed, J.~Neville, R.~A. Rossi, and N.~Duffield, ``Efficient graphlet
  counting for large networks,'' in \emph{ICDM'15}.\hskip 1em plus 0.5em minus
  0.4em\relax IEEE, 2015, pp. 1--10.

\bibitem{jain2017}
S.~Jain and C.~Seshadhri, ``A fast and provable method for estimating clique
  counts using tur\'{a}n's theorem,'' in \emph{WWW '17}, 2017, pp. 441--449.

\bibitem{finochi2015}
I.~Finocchi, M.~Finocchi, and E.~G. Fusco, ``Clique counting in mapreduce:
  Algorithms and experiments,'' \emph{JEA}, vol.~20, pp. 1.7:1--1.7:20, Oct.
  2015.

\bibitem{bordino2008}
I.~Bordino, D.~Donato, A.~Gionis, and S.~Leonardi, ``Mining large networks with
  subgraph counting,'' in \emph{ICDM'08}.\hskip 1em plus 0.5em minus
  0.4em\relax IEEE, 2008, pp. 737--742.

\bibitem{jha2015cliques}
M.~Jha, C.~Seshadri, and A.~Pinar, ``Path sampling: A fast and provable method
  for estimating 4-vertex subgraph counts,'' in \emph{WWW '15}, 2015, pp.
  495--505.

\bibitem{destefani2017}
\BIBentryALTinterwordspacing
L.~De~Stefani, E.~Terolli, and E.~Upfal, ``Tiered sampling: An efficient method
  for approximate counting sparse motifs in massive graph streams - extended
  materials,'' 2017. [Online]. Available:
  \url{http://bigdata.cs.brown.edu/tiered_sampling.pdf}
\BIBentrySTDinterwordspacing

\bibitem{mitzenmacher2005probability}
M.~Mitzenmacher and E.~Upfal, \emph{Probability and Computing: Randomization
  and Probabilistic Techniques in Algorithms and Data Analysis}.\hskip 1em plus
  0.5em minus 0.4em\relax Cambridge university press, 2017.

\bibitem{albert2002}
R.~Albert and A.-L. Barab{\'a}si, ``Statistical mechanics of complex
  networks,'' \emph{Reviews of modern physics}, vol.~74, no.~1, p.~47, 2002.

\bibitem{boldi2011}
P.~Boldi, M.~Rosa, M.~Santini, and S.~Vigna, ``Layered label propagation: A
  multiresolution coordinate-free ordering for compressing social networks,''
  in \emph{WWW'11}.\hskip 1em plus 0.5em minus 0.4em\relax ACM, 2011.

\bibitem{mislove2007}
A.~Mislove, M.~Marcon, K.~P. Gummadi, P.~Druschel, and B.~Bhattacharjee,
  ``Measurement and analysis of online social networks,'' in \emph{ACM
  SIGCOMM'07}, 2007, pp. 29--42.

\end{thebibliography}

\newpage
\clearpage
\newpage
\onecolumn
\appendices
\section{Additional theoretical results for \algomultione{}}\label{app:probability computation}
Before presenting the proof for the main analytical results for \algomultione{} discussed in Section~\ref{sec:multi}, we introduce some technical lemmas.

\begin{lemma}\label{lem:twocount}
	Let $\lambda\in \clisett{4}{t}$ with $\lambda = \{e_1,e_2,e_3,e_4,e_5,e_6\}$ using Figure~\ref{fig:4clidetection} as reference. Assume further, without loss of generality, that the edge $e_i$ is observed at $t_i$ (not necessarily consecutively) and that $t_6>\max\{t_i,1\leq i\leq5\}$. $\lambda$ can be observed by \algomultione{} at time $t_6$ either as a combination of triangle $T_1=\{e_1,e_2, or e_4\}$ and edges $e_3=(v,w)$ and $e_5=(v,z)$, or as a combination of triangle $T_2=\{e_1,e_3,e_5\}$ and edges $e_2=(u,w)$ and $e_4=(u,z)$.
\end{lemma}
\begin{proof}
	As presented in Algorithm~\ref{alg:ncliquecounter}, \algomultione{}  can detect $\lambda$ only when its last edge is observed on the stream (hence, $t_6$). When $e_6$ is observed, the algorithm first evaluates wether there is any triangle in $\Samp_\Delta^{(t_6)}$ that shares one of the two endpoint $u$ or $v$ form $e_6$. Since a this step (i.e., the execution of fucntion \textsc{Update4Cliques}) the triangle sample is jet to be updated based on the observation of $e_6$, the only triangle sub-structures of $\lambda$ which may have been observed on the stream, and thus included in $\Samp_\Delta^{(t_6)}$ are $T_1=\{e_1,e_2,e_4\}$ and $T_2=\{e_1,e_3,e_5\}$.
		If any of these is indeed in $\Samp_\Delta^{(t_6)}$, \algomultione{} proceeds to check wether the remaining two edges required to complete $\lambda$ (resp., $e_3,e_5$ for $T_1$, or $e_2,e_4$ for $T_2$) are in $\Samp_e$. $\lambda$ is thus observed either once or twice depending on which just one or both of these conditions are verified.
\end{proof}
\bigskip
\begin{lemma}\label{lem:probclique}
Let $\lambda\in \clisett{4}{t}$ with $\lambda = \{e_1,e_2,e_3,e_4,e_5,e_6\}$ using Figure~\ref{fig:4clidetection} as reference. Assume further, without loss of generality, that the edge $e_i$ is observed at $t_i$ (not necessarily consecutively) and that $t_6>\max\{t_i,1\leq i\leq5\}$. Let $t_{1,2,4}=\max{t_1,t_2,t_4,M_e+1}$. The probability $p_\lambda$ of $\lambda$ being observed on the stream by \algomulti{} using the triangle $T_1=\{e_1,e_2,e_4\}$ and the edges $e_3, e_5$, is computed by \textsc{ProbClique} as:
\begin{align*}\label{eq:probclique}
	p_\lambda = \frac{M_e}{t^M_{1,2,4}-1}\frac{M_e-1}{t^M_{1,2,4}-2}\minone{\frac{M_\Delta}{\ntris[t_6]}}p'
\end{align*}
where if $t_6\leq M_e$
\begin{equation*}
	p' = 1,
\end{equation*}
if $\min\{t_3,t_5\}>t_{1,2,4}$
\begin{equation*}
	p' = \frac{M_e}{t_6-1}\frac{M_e-1}{t_6-2},
\end{equation*}
if $\max\{t_3,t_5\}>t_{1,2,4}>\min\{t_3,t_5\}$
\begin{equation*}
	p'= \frac{M_e-1}{t_6-2}\frac{M_e-2}{t_{1,2,4}-3}\frac{t_{1,2,4}-1}{t_6-1}
\end{equation*}
otherwise
\begin{equation*}
	p'=\frac{M_e-2}{t_{1,2,4}-3}\frac{M_e-3}{t_{1,2,4}-4}\frac{t_{1,2,4}-1}{t_6-1}\frac{t_{1,2,4}-2}{t_6-2}.
\end{equation*}
\end{lemma}

\bigskip
\begin{proof}
Let us define the event $E_\lambda =\lambda$\emph{ is observed on the stream by \algomultione{} using triangle $T_1=\{e_1,e_2,e_4\}$ and edges $e_3,e_5$}. Further let $E_{T_1} = T_1\in \Samp_{\Delta}^{(t_6)}$, and $E_{3,5}= \{e_3,e_5\}\subseteq \Samp_e^{(t_6)}$.
Given the definition of \algomultione{} we have:
\begin{equation*}
	E_\lambda = E_{T_1} \wedge E_{3,5},
\end{equation*}
and hence:
\begin{equation*}
	p_\lambda = \Prob{E_\lambda} = \Prob{E_{T_1} \wedge E_{3,5}} = \Prob{E_{3,5}|E_{T_1}}\Prob{E_{T_1}}.
\end{equation*}
In oder to study $\Prob{E_{T_1}}$ we shall introduce event $E_{S(T_1)}=$``\emph{triangle $T_1$ is observed on the stream by \algomultione{}}''. From the definition of \algomultione{} we know that $T_1$ is observed on the stream iff when the last edge of $T_1$ is observed on the stream at $\max\{t_1,t_2,t_4\}$ the remaining to edges are in the edge sample. Applying Bayes's rule of total probability we have:
\begin{equation*}
	\Prob{E_{T_1}} = \Prob{E_{T_1}|E_{S(T_1)}}\Prob{E_{S(T_1)}}.
\end{equation*} 
and thus:
\begin{equation}\label{eq:unb2}
	p_\lambda = \Prob{E_{3,5}|E_{T_1}}\Prob{E_{T_1}|E_{S(T_1)}}\Prob{E_{S(T_1)}}.
\end{equation}
Let $t_{1,2,4}=\max\{t_1,t_2,t_4,M+1\}$, if order for $T_1$ to be observed by \algomultione{} it is required that when the last edge of $T_1$ is obseved on the stream at $t_{1,2,4}$ its two remaining edges are kept in $\Samp_e$. From Lemma~\ref{lem:reservoirhighorder} we have:
\begin{equation}\label{eq:t1seen}
	\Prob{E_{S(T_1)}} = \frac{M_e}{t_{1,2,4}-1}\frac{M_e-1}{t_{1,2,4}-2},
\end{equation}
and:
\begin{equation}\label{eq:t1stored}
	\Prob{E_{T_1}|E_{S(T_1)}}=\frac{M_e}{\tau^{t_6}}.
\end{equation}
Let us now consider $\Prob{E_{3,5}|E_{T_1}}$. While the content of $\Samp_\Delta$ itself does not influence the content of $\Samp_e$, the fact that $T_1$ is maintained in $\Samp\Delta^{(t_6)}$ implies that is has been observed on the stream at a previous time and hence, that two of its edges have been maintained in $\Samp_e$ \emph{at least} until the last of its edges has been observed on the stream. We thus have $\Prob{E_{3,5}|E_{T_1}}=\Prob{E_{3,5}|E_{S(T_1})}$. In order to study $p'=\Prob{E_{3,5}|E_{S(T_1})}$ it is necessary to distinguish the possible ($5!$) different arrival orders for edges $e_1,e_2,e_3,e_4$ and $e_5$. An efficient analysis we however reduce the number of cases to be considered to just four:
\begin{itemize}
\item $t_6\leq M+e$: in this case \emph{all} the edges observed on the stream up until $t_6$ are deterministically inserted in $\Samp_e$ and thus $p'=1$.
\item $\min\{t_3,t_5\}>t_{1,2,4}$: in this cases both edges $e_3$ and $e_5$ are observed after \emph{all} the edges composing $T_1$ have already been observed on the stream. As for any $t>t_{1,2,4}$ the event $E_{S(T_1)}$ does not imply that any of the edges of $T_1$ is \emph{still} in $\Samp_e$ we have:
\begin{equation*}
	p' = \Prob{E_{3,5}|E_{T_1}} = \Prob{\{e_3,e_5\}\subseteq \Samp_e^{(t_6)}},
\end{equation*}
and thus, from Lemma~\ref{lem:reservoirhighorder}:
\begin{equation*}
	p' = \frac{M_e}{t_6-1}\frac{M_e-1}{t_6-2}.
\end{equation*}
\item $\max\{t_3,t_5\}>t_{1,2,4}>\min\{t_3,t_5\}$: in this case only one of the edges $e_3,e_5$ is observed after all the edges in $T_1$ are observed. We need to therefore take into consideration that the two edges of $T_1$ are kept in $S_e$ until $t_{1,2,4}$. Let $e^M_{3,5}$ (resp., $e^m_{3,5}$) denote the last (resp., first) edge observed on the stream between $e_3$ and $e_5$.
\begin{align*}
	p'&= \Prob{e^M_{3,5}\in \Samp_e^{(t_6)}|e^m_{3,5}\in \Samp_e^{(t_6)}}\Prob{e^m_{3,5}\in \Samp_e^{(t_6)}},\\
	&= \frac{M_e-1}{t_6-2}\Prob{e^m_{3,5}\in \Samp_e^{(t_6)}|e^m_{3,5}\in \Samp_e^{(t_{1,2,4})}}\Prob{e^m_{3,5}\in \Samp_e^{(t_{1,2,4})}},\\
	&= \frac{M_e-1}{t_6-2}\frac{t_{1,2,4}-1}{t_6-1}\frac{M_e-2}{t_{1,2,4}-3}.
\end{align*}
\item $t_{1,2,4}>\max\{t_2,t_4\}$: in all the remaining cases both $e_3$ and $e_5$ are observed before the last edge of $T_1$ has been observed. Hence:
\begin{align*}
	p'&= \Prob{\{e_3,e_5\}\subseteq \Samp_e^{(t_6)}|\{e_3,e_5\}\subseteq \Samp_e^{(t_{1,2,4})}}\Prob{\{e_3,e_5\}\subseteq \Samp_e^{(t_{1,2,4})}},\\
	&= \frac{M_e-2}{t_{1,2,4}-3}\frac{M_e-3}{t_{1,2,4}-4}\frac{t_{1,2,4}-1}{t_6-1}\frac{t_{1,2,4}-2}{t_6-2}.
\end{align*}
\end{itemize}
The lemma follows combining the result for the values of $p'$ with~\eqref{eq:t1seen} and~\eqref{eq:t1stored} in~\eqref{eq:unb2}. 
\end{proof}
\bigskip
Using this result we can proceed to the proof of the unbiasedness of the estimator returned by \algomultione{}.
\bigskip
\begin{proof}[Proof of Theorem~\ref{thm:multiunbiased}]
From Lemma~\ref{lem:twocount} we have that \algomultione{} can detect any 4-clique $\lambda\in\clisett{4}{t}$ in exactly two ways: either using triangle $T_1$ and edges $e_3,e_5$, or by using triangle $T_2$ and edges $e_2,e_4$ (use Figure~\ref{fig:4clidetection} as a reference).

For each $\lambda\in\clisett{4}{t}$ let us consider the random variable $\delta_{\lambda_1}$ (resp. $\delta_{\lambda_2}$) which takes value $p^{-1}_{\lambda_1}/2$ (resp., $p^{-1}_{\lambda_2}/2$) if the 4-clique $\lambda$ is observed by \algomultione{} using triangle $T_1$ (resp., $T_2$) or zero otherwise. Let $p_{\lambda_1}$ (resp., $p_{\lambda_2}$) denote the probability of such event: we then have $\Exp{\delta_{\lambda_1}}=\Exp{\delta_{\lambda_2}}=1/2$. 
	
From Lemma~\ref{lem:probclique} we the estimator $\fcliest[t]$ computed using \algomultione{} has can be expressed as $\fcliest[t] = \sum_{\lambda\in \clisett{4}{t}} \left(
	\delta_{\lambda_1} + \delta_{\lambda_2} \right)$. From the previous discussion and by applying linearity of expectation we thus have:
\begin{equation*}
	\Exp{\fcliest[t]} = \Exp{\sum_{\lambda\in \clisett{4}{t}} \left(
	\delta_{\lambda_1} + \delta_{\lambda_2} \right)} =\sum_{\lambda\in \clisett{4}{t}} \left(
	\Exp{\delta_{\lambda_1}} + \Exp{\delta_{\lambda_2}} \right)= \sum_{\lambda\in \clisett{4}{t}} 1 = |\clisett{4}{t}|.
\end{equation*}

Finally, let $t^*$ denote the first step at for which the number of triangles seen exceeds $M_\Delta$. 
For $t \leq \min \{M_e, t^*\}$, the entire graph $G^{(t)}$ is maintained in $\Samp_e$ and all the triangles in $G^{(t)}$ are stored in $\Samp_\Delta$. Hence all the cliques in $G^{(t)}$ are deterministically observed by \algomulti \emph{in both ways} and we therefore have $\varkappa^{(t)}= |\clisett{4}{t}|$.
\end{proof}
\bigskip
\begin{lemma}\label{lem:sharingclique}
	Any pair $(\lambda,\gamma)$ of distinct 4-cliques in $G^{(t)}$ can share either one, three or no edges. If $\lambda$ and $\gamma$ share three edges, those three edges compose a triangle.
\end{lemma}
\begin{proof}
	Suppose that $\lambda$ and $\gamma$ share exactly two distinct edges. This implies that they share at least three distinct nodes, and thus must share the three edges connecting each pair out of said three nodes. This constitutes a contradiction. 
	Suppose instead that $\lambda$ and $\gamma$ share four or five edges while being distinct. This implies that they must share four vertices, hence they cannot be distinct cliques. This leads to a contradiction.
\end{proof}
\begin{figure}
\begin{center}
	\begin{tikzpicture}[line cap=round,line join=round,>=triangle 45,x=1cm,y=1cm, scale=1, every node/.style={scale=1}]
	\fill[line width=1.2pt,color=zzttqq,fill=zzttqq,fill opacity=0.10000000149011612] (0,3) -- (0,0) -- (3,0) -- cycle;
	
	\fill[line width=0pt,color=qqqqff,fill=qqqqff,fill opacity=0.1] (3,3) -- (0,0) -- (3,0) -- cycle;
	
	\draw [line width=1.2pt] (0,3)-- (0,0);
	
	\draw [line width=1.2pt] (0,0)-- (3,0);
	
	\draw [line width=1.2pt,dash pattern=on 1pt off 3pt] (3,0)-- (3,3);
	
	\draw [line width=1.2pt] (3,3)-- (0,3);
	
	\draw [line width=1.2pt] (3,3)-- (0,0);
	
	\draw [line width=1.2pt] (0,3)-- (3,0);
	
	\draw [line width=1.2pt] (6,0)-- (6,3);
	\draw [line width=1.2pt] (3,0)-- (6,0);
	\draw [line width=1.2pt] (3,0)-- (6,3);
	\draw [line width=1.2pt] (3,3)-- (6,3);
	\draw [line width=1.2pt] (3,3)-- (6,0);
	\begin{scriptsize}
	\draw [fill=xdxdff] (0,0) circle (2pt);
	\draw [fill=xdxdff] (0,3) circle (2pt);
	\draw [fill=qqqqff] (3,3) circle (2pt);
	\draw [fill=xdxdff] (3,0) circle (2pt);
	\draw [fill=xdxdff] (6,3) circle (2pt);
	\draw [fill=xdxdff] (6,0) circle (2pt);
	\draw[color=black] (-0.2,1.5) node {$e_{2}$};
	\draw[color=black] (1.5,-0.2) node {$e_{1}$};
	\draw[color=black] (3.2,1.5) node {$e^{*}$};
	\draw[color=black] (1.5,3.2) node {$e_{6}$};
	\draw[color=black] (0.95,2.3) node {$e_{4}$};
	\draw[color=black] (2.05,2.3) node {$e_{5}$};
	\draw[color=zzttqq] (0.68,1.5) node {$T_1$};
	\draw[color=qqqqff] (2.44,1.5) node {$T_2$};
	
	\draw[color=black] (4.5,-0.2) node {$g_{1}$};
	\draw[color=black] (6.2,1.5) node {$g_{3}$};
	\draw[color=black] (4.5,3.2) node {$g_{6}$};
	\draw[color=black] (4.05,2.3) node {$g_{4}$};
	\draw[color=black] (5.05,2.3) node {$g_{5}$};
	\end{scriptsize}
	\end{tikzpicture}
	\caption{Cliques sharing one edge.}
	\label{fig:share1edge}
	\end{center}
\end{figure}
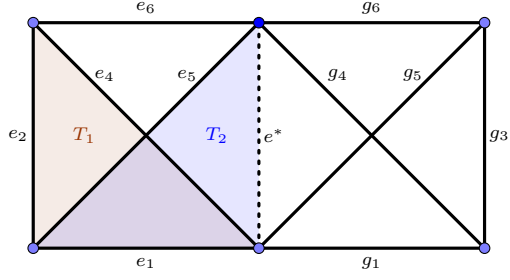

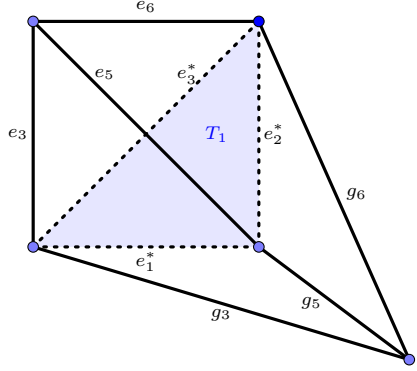
\begin{figure}
\begin{center}
	\begin{tikzpicture}[line cap=round,line join=round,>=triangle 45,x=1cm,y=1cm, scale=1, every node/.style={scale=1}]
	
	\fill[line width=0pt,color=qqqqff,fill=qqqqff,fill opacity=0.1] (3,3) -- (0,0) -- (3,0) -- cycle;
	
	\draw [line width=1.2pt] (0,3)-- (0,0);
	
	\draw [line width=1.2pt,dash pattern=on 1pt off 3pt] (0,0)-- (3,0);
	
	\draw [line width=1.2pt,dash pattern=on 1pt off 3pt] (3,0)-- (3,3);
	
	\draw [line width=1.2pt] (3,3)-- (0,3);
	
	\draw [line width=1.2pt,dash pattern=on 1pt off 3pt] (3,3)-- (0,0);
	
	\draw [line width=1.2pt] (0,3)-- (3,0);
	
	\draw [line width=1.2pt] (3,0)-- (5,-1.5);
	\draw [line width=1.2pt] (3,3)-- (5,-1.5);
	\draw [line width=1.2pt] (0,0)-- (5,-1.5);
	\begin{scriptsize}
	\draw [fill=xdxdff] (0,0) circle (2pt);
	\draw [fill=xdxdff] (0,3) circle (2pt);
	\draw [fill=qqqqff] (3,3) circle (2pt);
	\draw [fill=xdxdff] (3,0) circle (2pt);
	\draw [fill=xdxdff] (5,-1.5) circle (2pt);
	\draw[color=black] (-0.2,1.5) node {$e_{3}$};
	\draw[color=black] (1.5,-0.2) node {$e_{1}^*$};
	\draw[color=black] (3.2,1.5) node {$e_2^{*}$};
	\draw[color=black] (1.5,3.2) node {$e_{6}$};
	\draw[color=black] (0.95,2.3) node {$e_{5}$};
	\draw[color=black] (2.05,2.3) node {$e_{3}^*$};
	\draw[color=qqqqff] (2.44,1.5) node {$T_1$};
	\draw[color=black] (2.5,-0.9) node {$g_{3}$};
	\draw[color=black] (3.7,-0.75) node {$g_{5}$};
	\draw[color=black] (4.3,0.7) node {$g_{6}$};

	\end{scriptsize}
	\end{tikzpicture}
	\caption{Cliques sharing three edges.}
	\label{fig:share3edges}
	\end{center}
\end{figure}

We now proceed to the proof of the bound on the variance of \algomultione{} in Theorem~\ref{thm:variancemulti}. We remark that the bound presented here is loose as it sacrifices its precision for presentation purposes. A stronger bound would require a case by case analysis of all the 12! possible relative ordering according to which the edges of any pair of 4-cliques in $G^{(t)}$ are observed. Our proof technique simplifies the analysis by grouping these cases into macro-cases with a resulting loss in tightness. 
\bigskip
\begin{proof}[Proof of Theorem~\ref{thm:variancemulti}]
	Assume $|\clisett{4}{t}|>0$, otherwise \algomultione{} estimation is
  deterministically correct and has variance 0 and the thesis holds. For each $\lambda\in\clisett{4}{t}$ let $\lambda = \{e_1,e_2,e_3,e_4,e_5,e_6\}$, withouth loss of generality let us assume the edges are disposed as in Figure~\ref{fig:4clidetection}.  Assume further, without loss of generality, that the edge $e_i$ is observed at $t_i$ (not necessarily consecutively) and that $t_6>\max\{t_i,1\leq i\leq5\}$. Let $t_{1,2,4}=\max\{t_1,t_2,t_4,M_e+1\}$. Let us consider the random variable $\delta_{\lambda_1}$ (resp. $\delta_{\lambda_2}$) which takes value $p^{-1}_{\lambda_1}/2$ (resp., $p^{-1}_{\lambda_2}/2$) if the 4-clique $\lambda$ is observed by \algomultione{} using triangle $T_1 = \{e_1,e_2,e_4\}$ (resp., $T_2$) and edges $e_3,e_5$ (resp., $e_2,e_4$) or zero otherwise. Let $p_{\lambda_1}$ (resp., $p_{\lambda_2}$) denote the probability of such event.
  
   Since, from Lemma~\ref{lem:probclique} we know:
  \begin{align*}
  	 \Var{\delta_{\lambda_1}} &= \frac{p_{\lambda_1}^{-1}}{4}-\frac{1}{4}\leq \frac{1}{4}\left(\frac{\ntris[t]}{M_\Delta}\prod_{i=0}^3 \frac{t-1-i}{M_e-i}-1\right),\\
    \Var{\delta_{\lambda_2}} &= \frac{p_{\lambda_2}}{4}^{-1}-\frac{1}{4}\leq \frac{1}{4}\left(\frac{\ntris[t]}{M_\Delta}\prod_{i=0}^3 \frac{t-1-i}{M_e-i}-1\right).\\
  \end{align*}  
  we have:
  \begin{align}\label{eq:variacemul1n1}
    \Var{\varkappa^{(t)}} &=
    \Var{\sum_{\lambda\in\clisett{4}{t}} \delta_{\lambda_1}+\delta_{\lambda_2}}=
    \sum_{\lambda\in \clisett{4}{t}} \sum_{\gamma \in
    \clisett{4}{t}}\sum_{i\in\{1,2\}}\sum_{j\in\{1,2\}}\text{Cov}\left[\delta_{\lambda_i},\delta_{\gamma_j}\right]
    \nonumber\\
    &= \sum_{\lambda\in \lambda\in \clisett{4}{t}}\left(
    \Var{\delta_{\lambda_1}}+\Var{\delta_{\lambda_2}}\right) +\sum_{\lambda\in \clisett{4}{t}}\left(
    \Cov{\delta_{\lambda_1},\delta_{\lambda_2}}+\Cov{\delta_{\lambda_2},\delta_{\lambda_1}}\right)\nonumber \\
    &+\sum_{\substack{\lambda\in \clisett{4}{t}\\\lambda \neq \gamma}} \left(\Cov{\delta_{\lambda_1},\delta_{\gamma_1}}+\Cov{\delta_{\lambda_1},\delta_{\gamma_2}}+\Cov{\delta_{\lambda_2},\delta_{\gamma_1}}+\Cov{\delta_{\lambda_2},\delta_{\gamma_2}}\right)
    \nonumber\\
   &= \frac{|\clisett{4}{t}|}{2}\left(\frac{\ntris[t]}{M_\Delta}\prod_{i=0}^3 \frac{t-1-i}{M_e-i}-1\right) +\sum_{\lambda\in \clisett{4}{t}}\left(
    \Cov{\delta_{\lambda_1},\delta_{\lambda_2}}+\Cov{\delta_{\lambda_2},\delta_{\lambda_1}}\right)\nonumber \\
    &+\sum_{\substack{\lambda\in \clisett{4}{t}\\\lambda \neq \gamma}} \left(\Cov{\delta_{\lambda_1},\delta_{\gamma_1}}+\Cov{\delta_{\lambda_1},\delta_{\gamma_2}}+\Cov{\delta_{\lambda_2},\delta_{\gamma_1}}+\Cov{\delta_{\lambda_2},\delta_{\gamma_2}}\right)\\
    &\leq \frac{|\clisett{4}{t}|}{2}\left(\frac{\ntris[t]}{M_\Delta}c\left(\frac{t-1}{M_e}\right)^4-1\right) +\sum_{\lambda\in \clisett{4}{t}}\left(
    \Cov{\delta_{\lambda_1},\delta_{\lambda_2}}+\Cov{\delta_{\lambda_2},\delta_{\lambda_1}}\right)\nonumber \\
    &+\sum_{\substack{\lambda\in \clisett{4}{t}\\\lambda \neq \gamma}} \left(\Cov{\delta_{\lambda_1},\delta_{\gamma_1}}+\Cov{\delta_{\lambda_1},\delta_{\gamma_2}}+\Cov{\delta_{\lambda_2},\delta_{\gamma_1}}+\Cov{\delta_{\lambda_2},\delta_{\gamma_2}}\right)
    \enspace.
  \end{align}

We now proceed to analyze the various covariance terms appearing in~\eqref{eq:variacemul1n1}. In the following we refer to 

\begin{itemize}
	\item The second summation in~\eqref{eq:variacemul1n1}, $\sum_{\lambda\in \clisett{4}{t}}\left(
    \Cov{\delta_{\lambda_1},\delta_{\lambda_2}}+\Cov{\delta_{\lambda_2},\delta_{\lambda_1}}\right)$, concerns the sum of the covariances of pairs of random variables each corresponding to one of the two possible ways of detecting a 4-clique using \algomultione{}. Let us consider one single element of the summation:
    \begin{equation*}
    	\Cov{\delta_{\lambda_1},\delta_{\lambda_2}} = \Exp{\delta_{\lambda_1}\delta_{\lambda_2}} -\Exp{\delta_{\lambda_1}}\Exp{\delta_{\lambda_2}} = \Exp{\delta_{\lambda_1}\delta_{\lambda_2}} -1/4.
    \end{equation*}
    Let us now focus on $\Exp{\delta_{\lambda_1}\delta_{\lambda_2}}$, according to the definition of $\delta_{\lambda_1}$ and $\delta_{\lambda_2}$ we have:
    \begin{align*}
    	\Exp{\delta_{\lambda_1}\delta_{\lambda_2}} &= \frac{p_{\lambda_1}^{-1}p_{\lambda_2}^{-1}}{4} \Prob{\delta_{\lambda_1}=p^{-1}_{\lambda_1}\wedge\delta_{\lambda_2}=p^{-1}_{\lambda_2}} \\
    	&= \frac{p_{\lambda_1}^{-1}p_{\lambda_2}^{-1}}{4} \Prob{\delta_{\lambda_1}=p^{-1}_{\lambda_1}|\delta_{\lambda_2}=p^{-1}_{\lambda_2}}\Prob{\delta_{\lambda_2}=p^{-1}_{\lambda_2}}\\
    	&\leq \frac{p_{\lambda_1}^{-1}p_{\lambda_2}^{-1}}{4} \Prob{\delta_{\lambda_2}=p^{-1}_{\lambda_2}}\\
    	&\leq \frac{p_{\lambda_1}^{-1}p_{\lambda_2}^{-1}}{4} p_{\lambda_2}\\
    	&\leq \frac{p_{\lambda_1}^{-1}}{4}.
    \end{align*}
    
    We can therefore conclude:
    \begin{equation}\label{eq:varmulti1p1}
    	\sum_{\lambda\in \clisett{4}{t}}\left(
    \Cov{\delta_{\lambda_1},\delta_{\lambda_2}}+\Cov{\delta_{\lambda_2},\delta_{\lambda_1}}\right) \leq \frac{|\clisett{4}{t}|}{2}\left(\frac{\ntris[t]}{M_\Delta}\prod_{i=0}^3 \frac{t-1-i}{M_e-i}-1\right) \leq \frac{|\clisett{4}{t}|}{2}\left(\frac{\ntris[t]}{M_\Delta}c\left(\frac{t-1}{M_e}\right)^4-1\right)
    \end{equation}
   \item The third summation in~\eqref{eq:variacemul1n1}, includes the covariances of  all $|\clisett{4}{t}|\left(2|\clisett{4}{t}|-1\right)$ unordered pairs of random variables corresponding each to one of the two possible ways of counting distinct 4-cliques in $\clisett{4}{t}$. In order to provide a significant bound it is necessary to divide the possible pairs of 4-cliques depending on how many edges they share (if any). From Lemma~\ref{lem:sharingclique} we have that any pair of 4-cliques $\lambda$ and $\gamma$ can share either one, three or no edges. In the remainder of our analysis we shall distinguishing three group of pairs of 4-cliques based on how many edges they share. In the following, we present, without loss of generality, bounds for $\Cov{\delta_{\lambda_1},\delta_{\gamma_2}}$. The results steadily holds for the other possible combinations $\Cov{\delta_{\lambda_1},\delta_{\gamma_1}},\ \Cov{\delta_{\lambda_2},\delta_{\gamma_1}},\ \Cov{\delta_{\lambda_2},\delta_{\gamma_2}},\ \Cov{\delta_{\gamma_1},\delta_{\lambda_1}},\ \Cov{\delta_{\gamma_1},\delta_{\lambda_2}},\ \Cov{\delta_{\gamma_2},\delta_{\lambda_1}}$, and $\Cov{\delta_{\gamma_2},\delta_{\lambda_2}}$
   \begin{enumerate}
   	\item $\lambda$ and $\gamma$ do not share any edge:
   	\begin{align*}
    	\Exp{\delta_{\lambda_1}\delta_{\gamma_2}} &= \frac{p_{\lambda_1}^{-1}p_{\gamma_2}^{-1}}{4} \Prob{\delta_{\lambda_1}=p^{-1}_{\lambda_1} \wedge \delta_{\gamma_2}=p^{-1}_{\gamma_2}} \\
    	&= \frac{p_{\lambda_1}^{-1}p_{\gamma_2}^{-1}}{4} \Prob{\delta_{\lambda_1}=p^{-1}_{\lambda_1}|\delta_{\gamma_2}=p^{-1}_{\gamma_2}}\Prob{\delta_{\lambda_2}=p^{-1}_{\lambda_2}}\\
    \end{align*}
    The term $\Prob{\delta_{\lambda_1}=p^{-1}_{\lambda_1}|\delta_{\gamma_2}=p^{-1}_{\gamma_2}}$ denotes the probability of \algomultione{} observing $\lambda$ using $T_1$ and edges $e_3$ and $e_5$ conditioned of the fact that $\gamma$ was observed by the algorithm using $T_3$ and edges $g_3$ and $g_5$. Note that as $\lambda$ and $\gamma$ do not share any edge, no edge of $\gamma$ will be used by \algomultione{} to detect $\lambda$. Rather, if any edge of $\gamma$ is included in $\S_e$ or if $T_3$ is included in $\Samp_\Delta$, this would lessen the probability of \algomultione{} detecting $\lambda$ using $T1,e_3$ and $e_5$ as some of space in $\Samp_e$ or $\Samp_\Delta$ may be occupied by edges or triangle sub-structures for $\gamma$. Therefore we have $ \Prob{\delta_{\lambda_1}=p^{-1}_{\lambda_1}|\delta_{\gamma_2}=p^{-1}_{\gamma_2}}\leq  \Prob{\delta_{\lambda_1}=p^{-1}_{\lambda_1}}$ and thus:
    \begin{align*}
    	\Exp{\delta_{\lambda_1}\delta_{\gamma_2}} &= \frac{p_{\lambda_1}^{-1}p_{\gamma_2}^{-1}}{4} \Prob{\delta_{\lambda_1}=p^{-1}_{\lambda_1}|\delta_{\gamma_2}=p^{-1}_{\gamma_2}}\Prob{\delta_{\lambda_2}=p^{-1}_{\lambda_2}}\\
    	&\leq \frac{p_{\lambda_1}^{-1}p_{\gamma_2}^{-1}}{4} \Prob{\delta_{\lambda_1}=p^{-1}_{\lambda_1}}\Prob{\delta_{\lambda_2}=p^{-1}_{\lambda_2}}\\
    	&\leq \frac{p_{\lambda_1}^{-1}p_{\gamma_2}^{-1}}{4} p_{\lambda_1}p_{\lambda_2}\\
    	&\leq \frac{1}{4}
    \end{align*}
    We therefore have $\Cov{\delta_{\lambda_1},\delta_{\gamma_2}}\leq 0$. Hence we can conclude that the contribution of the covariances of the pairs of random variables corresponding to 4-cliques that do not share any edge to the summation in~\eqref{eq:variacemul1n1} is less or equal to zero.
    \item $\lambda$ and $\gamma$ share exactly one edge $e^*=\lambda \cap \gamma$ as shown in Figure~\ref{fig:share1edge}  Let us consider the event $E*=$``$e^*\cap T_1\cap E^{(t_{1,2,4}-1)}\in \Samp_{e}^{t_{1,2,4}}$, and $e^*\cap \{e_3,e_5\}\cap E^{(t_6-1)}\in S_e^{(t_6)}$''. Clearly $\Prob{\delta_{\lambda_1}=p^{-1}_{\lambda_1}|\delta_{\gamma_2}=p^{-1}_{\gamma_2}}\leq \Prob{\delta_{\lambda_1}=p^{-1}_{\lambda_1}|E*}$. Recall from Lemma~\ref{lem:probclique} that $\Prob{\delta_{\lambda_1}|E^*} = \Prob{\{e_3,e_5\}\subseteq \Samp_e^{(t_6)}|E^*}\Prob{T_1 \in \Samp_\Delta^{(t_6)}|E*}\Prob{S(T_1)|E^*}$, where $S(T_1)$ denotes the event ``$T_1$ is observed on the stream by \algomultione{}''. By applying the law of total probability e have that $\Prob{T_1 \in \Samp_\Delta^{(t_6)}|E*}\leq \Prob{T_1 \in \Samp_\Delta^{(t_6)}}$. The remaining two terms are influenced differently depending on whether $e^*\in T_1$ or $e^*\in\{e_3,e_5\}$
    \begin{itemize}
    \item if $e^*\in T_1$: we then have $\Prob{\{e_3,e_5\}\subseteq \Samp_e^{(t_6)}|E^*}\leq \Prob{\{e_3,e_5\}\subseteq \Samp_e^{(t_6)}}$. This follows from the properties of the reservoir sampling scheme as the fact that the edge $e^*$ is in $\Samp_e$ means that one unit of the available memory space required to hold $e_3$ or $e_5$ is occupied, at least for some time, by $e^*$.
    If $e^*$ is the last edge of $T_1$ observed in the stream we then have:
    \begin{align*}
    	\Prob{S(T_1)|E^*} &=\Prob{\{e_1,e_2,e_4\}\setminus \{e^*\} \in\Samp_e^({t_1,2,4})|E^*}\\
    	&=\Prob{\{e_1,e_2,e_4\}\setminus \{e^*\} \in\Samp_e^({t_1,2,4})}\\
    	&\leq \frac{M_e}{t_{1,2,4}-1}\frac{M_e-1}{t_{1,2,4}-2}
    \end{align*}
    Suppose instead that $e*$ is \emph{not} the last edge of $T_1$. Assume further , without loss of generality, that $t_2>\max\{t_1,t_4\}$.  \algomultione{} observes $T_1$ iff $\{e_1,e_4\} \in\Samp_e^({t_2})$. As in $E^*$ we assume that once observed $e^*$ is \emph{always} maintained in $\Samp_e$ until $t_{1,2,4}$ we have:
    \begin{align*}
    	\Prob{S(T_1)|E^*} &=\Prob{\{e_1,e_4\}\setminus \{e^*\} \in\Samp_e^({t_1,2,4})|E^*}\\
    	&\leq \frac{M_e-1}{t_{1,2,4}-2}
    \end{align*}
    \item if $e^*\in \{e_3,e_5\}$: we then have $\Prob{S(T_1)|E^*}\leq \Prob{S(T_1)}$. This follows from the properties of the reservoir sampling scheme as the fact that the edge $e^*$ is in $\Samp_e$ means that one unit of the available memory space required to hold the first two edges of $T_1$ until $t_{1,2,4}$  is occupied, at least for some time, by $e^*$. Further, using Lemma~\ref{lem:probclique} we have:
    \begin{itemize}
    	\item if $t_6\leq M_e$, then $\Prob{\{e_3,e_5\}\subseteq \Samp_e^{(t_6)}|E^*}=1$
		\item if $\min\{t_3,t_5\}>t_{1,2,4}$, then $\Prob{\{e_3,e_5\}\subseteq \Samp_e^{(t_6)}|E^*}\geq \frac{M_e}{t_6-1}$,
	\item if $\max\{t_3,t_5\}>t_{1,2,4}>\min\{t_3,t_5\}$, then $\Prob{\{e_3,e_5\}\subseteq \Samp_e^{(t_6)}|E^*}\geq \max\{\frac{M_e-1}{t_6-2},\frac{M_e-2}{t_{1,2,4}-3}\frac{t_{1,2,4}-1}{t_6-1}\}$,
	\item otherwise $\Prob{\{e_3,e_5\}\subseteq \Samp_e^{(t_6)}|E^*}\geq 
	\frac{M_e-2}{t_{1,2,4}-3}\frac{t_{1,2,4}-1}{t_6-1}$.
    \end{itemize}
    \end{itemize}
    Putting together these various results we have that $p_\lambda^{(-1)}\Prob{\delta_{\lambda_1}=p^{-1}_{\lambda_1}|\delta_{\gamma_2}=p^{-1}_{\gamma_2}}\leq p_\lambda^{(-1)}\Prob{\delta_{\lambda_1}=p^{-1}_{\lambda_1}|E*}\leq c\frac{t_6-1}{M_e}$ with $c=s$. Hence we have $\Exp{\delta_{\lambda_1}\delta_{\gamma_2}}\leq \frac{c}{4}\frac{t_6-1}{M_e}\leq \frac{c}{4}\frac{t-1}{M_e}$. We can thus bound the contribution to the third component of~\eqref{eq:variacemul1n1} given by the pairs of random variables corresponding to 4-cliques that share one edge as:
    \begin{equation}\label{eq:varm1p2}
    	2a^{(t)}\left(c\frac{t-1}{M_e}-1\right),
    \end{equation}
    where $a^{(t)}$ denotes the number of unordered pairs of 4-cliques which share one edge in $G^{(t)}$.
    
    \item $\lambda$ and $\gamma$ share three edges $\{e^*_1,e^*_2,e^*_3\}$ which form a triangle sub-structure for both $\lambda$ and $\gamma$. Let us refer to Figure~\ref{fig:share3edges}, without loss of generality let $T_1$ denote the triangle shared between the two cliques. We distinguish the kind of pairs for the random variables $\delta_{\lambda_i}$ and $\delta_{\gamma_j}$ cases:
		\begin{itemize} 
		\item $\delta_{\lambda_i}= p^{-1}_{\lambda_i}$ if $T_1\in \Samp_\Delta^{t_6-1} \wedge \{e_3,e_5\}\subseteq \Samp_e^{t_6-1}$ and $\delta_{\gamma_j}= p^{-1}_{\gamma_j}$ if $T_1\in \Samp_\Delta^{t_\gamma-1} \wedge \{g_3,g_5\}\subseteq \Samp_e^{t_\gamma-1}$, where $t_\gamma$ denotes the time step at which the last edge of $\gamma$ is observed. This is the case for which the random variables $\delta_{\lambda_i}$ and  $\delta_{\gamma_j}$ corresponds to \algosingle observing $\lambda$ and $\gamma$ using the \emph{shared triangle} $T_1$. Let us consider the event $E*=$``$T_1\in\Samp_\Delta^{(t_6)}$. Clearly $\Prob{\delta_{\lambda_i}=p^{-1}_{\lambda_i}|\delta_{\gamma_j}=p^{-1}_{\gamma_j}}\leq \Prob{\delta_{\lambda_i}=p^{-1}_{\lambda_i}|E*}$.
		In this case we have $\Prob{T_1\in\Samp^{(t_6)}|E^*}\leq 1$, while $\Prob{\{e_3,e_5\}\in\Samp_e^{(t_6)}|E^*}\leq Prob{\{e_3,e_5\}\in\Samp_e^{(t_6)}}$. This second fact follows from the properties of the reservoir sampling scheme as the fact that the edges $e^*_1$, $e^*_2$ and $e^*_3$ are in $\Samp_e$ \emph{at least} for the time required for $T_1$ to be observed, means that at least two unit of the available memory space required to hold the edges $e_3,e_5$  are occupied, at least for some time. Putting together these various results we have that $p_{\lambda_i}^{(-1)}Prob{\delta_{\lambda_i}=p^{-1}_{\lambda_i}|\delta_{\gamma_j}=p^{-1}_{\gamma_j}}\leq p_{\lambda_i}^{(-1)}\Prob{\delta_{\lambda_i}=p^{-1}_{\lambda_i}|E*}\leq c\left(\frac{t_6-1}{M_e}\right)^2\frac{\tau^{(t)}}{\Samp_\Delta}$. Hence we have $\Exp{\delta_{\lambda_i}\delta_{\gamma_j}}\leq \frac{c}{4}\left(\frac{t_6-1}{M_e}\right)^2\frac{\tau^{(t)}}{\Samp_\Delta}$ and $\Cov{\delta_{\lambda_i},\delta_{\gamma_j}}\leq \frac{c}{4}\left(\frac{t_6-1}{M_e}\right)^2\frac{\tau^{(t)}}{\Samp_\Delta}-\frac{1}{4}$.
		\item in all the remaining case, then the random variables $\delta_{\lambda_i}$ and  $\delta_{\gamma_j}$ corresponds to \algosingle not observing $\lambda$ and $\gamma$ using the shared triangle $T_1$ for both of them. Let $T^*$ denote the triangle sub-structure used by \algomultione{} to count $\lambda$ with respect to $\delta_{\lambda_i}$. Let us consider the event $E^*=$``$\{e^*_1,e^*_2,e^*_3\}\cap T_1\cap E^{(t_{1,2,4}-1)} \in \Samp_{e}^{t_{1,2,4}}$, $T_1\in\Samp_\Delta^{(t_6)}$ unless one of its edges is the last edge of $T^*$ observed on the stream, and  $\{e^*_1,e^*_2,e^*_3\}\cap \{e_3,e_5\}\cap E^{(t_{6}-1)}\in S_e^{(t_6)}$ if $e^*\in\{e_3,e_5\}$'', where $E^{(t)}$ denotes the set of edges observed up until time $t$ included. Clearly $\Prob{\delta_{\lambda_i}=p^{-1}_{\lambda_i}|\delta_{\gamma_j}=p^{-1}_{\gamma_j}}\leq \Prob{\delta_{\lambda_i}=p^{-1}_{\lambda_i}|E*}$. Note that in this case $|\{e^*_1,e^*_2,e^*_3\}\cap T_1|+|\{e^*_1,e^*_2,e^*_3\}\cap \{e_3,e_5\}|$. By analyzing $\Prob{\delta_{\lambda_i}=p^{-1}_{\lambda_i}|E*}$ in this case using similar steps as the ones described for the other sub-cases we have $p_{\lambda_i}^{(-1)}\Prob{\delta_{\lambda_i}=p^{-1}_{\lambda_i}|\delta_{\gamma_j}=p^{-1}_{\gamma_j}}\leq p_{\lambda_i}^{(-1)}\Prob{\delta_{\lambda_i}=p^{-1}_{\lambda_i}|E*}\leq c\left(\frac{t_6-1}{M_e}\right)^3$. Hence we have $\Exp{\delta_{\lambda_i}\delta_{\gamma_j}}\leq \frac{c}{4}\left(\frac{t_6-1}{M_e}\right)^3$ and $\Cov{\delta_{\lambda_i},\delta_{\gamma_j}}\leq \frac{c}{4}\left(\frac{t_6-1}{M_e}\right)^3-\frac{1}{4}$.
		\end{itemize}
		 We can thus bound the contribution to the third component of~\eqref{eq:variacemul1n1} given by the pairs of random variables corresponding to 4-cliques that share three edges as:
    \begin{equation}\label{eq:varm1p3}
    	2b^{(t)}\left(c\left(\frac{t-1}{M_e}\right)^2\left(\frac{1}{4}\frac{\tau^{(t)}}{M_\Delta} + \frac{3}{4}\frac{t-1}{M_e} \right)-1\right),
    \end{equation}
    where $b^{(t)}$ denotes the number of unordered pairs of 4-cliques which share three edges in $G^{(t)}$.
\end{enumerate}
\end{itemize} 
The Theorem follows by combining~\eqref{eq:varmulti1p1},~\eqref{eq:varm1p2} and~\eqref{eq:varm1p3} in~\eqref{eq:variacemul1n1}.
\end{proof}
\newpage
\section{Additional theoretical results for \algomultitwo{}}\label{app:probability computation2}
In this section we present the main analytical results for \algomultitwo{} discussed in Section~\ref{sec:algotwo}.
In order to simplify the presentation we use the following notation:
\begin{align*}
	t^M_{1,2,\ldots,i} &\triangleq \max\{t_1,t_2,\ldots,t_i\}\\
	t^m_{1,2,\ldots,i} &\triangleq \min\{t_1,t_2,\ldots,t_i\}
\end{align*} 
\begin{lemma}\label{lem:twocounttwo}
	Let $\lambda\in \clisett{4}{t}$ with $\lambda = \{e_1,e_2,e_3,e_4,e_5,e_6\}$ using Figure~\ref{fig:4clidetection} as reference. Assume further, without loss of generality, that the edge $e_i$ is observed at $t_i$ (not necessarily consecutively) and that $t_6>\max\{t_i,1\leq i\leq5\}$. $\lambda$ can be observed by \algomultitwo{} at time $t_6$  only by a combination of two triangles $T_1=\{e_1,e_2, e_4\}$ and $T_2=\{e_1,e_3, e_5\}$. 
\end{lemma}
\begin{proof}
	\algomultitwo{}  can detect $\lambda$ only when its last edge is observed on the stream (hence, $t_6$). When $e_6$ is observed, the algorithm first evaluates wether there are two triangles in  $\Samp_\Delta^{(t_6)}$, that share two endpoints and the other endpoints are $u$ and $v$ respectively. Since a this step (i.e., the execution of function \textsc{Update4Cliques}) the triangle sample is jet to be updated based on the observation of $e_6$, the only triangle sub-structures of $\lambda$ which may have been observed on the stream, and thus included in $\Samp_\Delta^{(t_6)}$ are $T_1=\{e_1,e_2,e_4\}$ and $T_2=\{e_1,e_3,e_5\}$. $\lambda$ is thus observed if and only if both of the triangles $T_1$ and $T_2$ are found in $\Samp_\Delta^{(t_6)}$.
\end{proof}

\begin{lemma}\label{lem:probclique2}
Let $\lambda\in \clisett{4}{t}$ with $\lambda = \{e_1,e_2,e_3,e_4,e_5,e_6\}$ using Figure~\ref{fig:4clidetection} as reference. Assume further, without loss of generality, that the edge $e_i$ is observed at $t_i$ (not necessarily consecutively) and that $t_6>\max\{t_i,1\leq i\leq5\}$. The probability $p_\lambda$ of $\lambda$ being observed on the stream by \algomultitwo{}, is computed by \textsc{ProbClique} as:
\begin{align*}\label{eq:probclique}
	p_\lambda = &\minone{\frac{M_e}{t^M_{1,3,5}-1}\frac{M_e-1}{t^M_{1,3,5}-2}}\minone{\frac{M_\Delta}{\ntris[t_6]}\frac{M_\Delta-1}{\ntris[t_6]-1}}p'
\end{align*}
where:
\begin{equation*}
	p'=\begin{cases}
		1, &if\ t^M_{1,2,4}\leq M_e\\
		\frac{M_e-2}{t_1-3}\frac{M_e-3}{t_1-4}, &if\ t_1>t^M_{2,3,4,5}\\
		\frac{M_e-2}{t_1-3}\frac{M_e-3}{t_1-4}, &if\ t^M_{3,5}>t_1>\max\{t^m_{3,5},t_2,t_4\}\\
		\frac{M_e-2}{t^M_{2,4}-3}, &if\ t^M_{3,5}>t^M_{2,4}>\max\{t^m_{3,5},t^m_{2,4},t_1\}\\
		\frac{M_e}{t_1-1}\frac{M_e-1}{t_1-2}, &if\ t^m_{3,5}>t_1>t^M_{2,4}\\
		\frac{M_e-1}{t^M_{2,4}-2} &if\ t^m_{3,5}>t^M_{2,4}>t_1\\
		\frac{M_e-1}{t^M_{2,4}-1}\frac{M_e-2}{t_1-3}\frac{t_1-1}{t_2-1}, &if\ t^M_{2,4}>t_1>\max\{M_e,t^m_{2,4},t^M_{3,5}\}\\
		\frac{M_e-1}{t^M_{2,4}-1}\frac{M_e}{t_2-1} &if\ t^M_{2,4}>M_e>t_1>\max\{t^m_{2,4},t^M_{3,5}\}\\
		\frac{t_3-1}{t^M_{2,4}-1}\frac{t_3-2}{t^M_{2,4}-2}\frac{M_e-2}{t^M_{3,5}-3},&if\ t^M_{2,4}>t^M_{3,5}>\max\{ M_e,t^m_{3,5},t^m_{2,4},t_1\}\\
		\frac{M_e}{t^M_{2,4}-1}\frac{M_e-1}{t^M_{2,4}-2},&if\ t^M_{2,4}>M_e>t^M_{3,5}>\max\{t^m_{3,5},t^m_{2,4},t_1\}\\
		\frac{M_e}{t^M_{2,4}-1}\frac{M_e-1}{t^M_{2,4}-2},&if\ t^m_{2,4}>t_1>t^M_{3,5}\\
		\frac{M_e-1}{t^M_{2,4}-2}\frac{t^M_{3,5}-1}{t^M_{2,4}-1},&if\ t^m_{2,4}>t^M_{3,5}>\max\{M,t_1\}\\
		\frac{M_e-1}{t^M_{2,4}-2}\frac{M_e}{t^M_{2,4}-1},&if\ t^m_{2,4}>M_e>t^M_{3,5}>t_1
			\end{cases}
\end{equation*}

\end{lemma}
\begin{proof}
Let us define the event $E_\lambda =\lambda$\emph{ is observed} on the stream by \algomultitwo{} using triangle $T_1=\{e_1,e_2,e_4\}$ and triangle $T_2=\{e_1,e_3,e_5\}$.
Given the definition of \algomultitwo{} we have:
\begin{equation*}
	E_\lambda = E_{T_1} \wedge E_{T_2},
\end{equation*}

and hence:
\begin{equation*}
\label{eq:prob2}
	p_\lambda = \Prob{E_\lambda} = \Prob{E_{T_1} \wedge E_{T_2}} = \Prob{E_{T_1}|E_{T_2}}\Prob{E_{T_2}}.
\end{equation*}
In oder to study $\Prob{E_{T_2}}$ we shall introduce event $E_{S(T_2)}=$``\emph{triangle $T_2$ is observed on the stream by \algomultitwo{}}''. From the definition of \algomultitwo{} we know that $T_2$ is observed on the stream iff when the last edge of $T_2$ is observed on the stream at $\max\{t_1,t_3,t_5\}$ the remaining to edges are in the edge sample. Applying Bayes's rule of total probability we have:
\begin{equation*}
	\Prob{E_{T_2}} = \Prob{E_{T_2}|E_{S(T_2)}}\Prob{E_{S(T_2)}}.
\end{equation*} 
and thus:
\begin{equation}
	p_\lambda = \Prob{E_{T_1}|E_{T_2}}\Prob{E_{T_2}|E_{S(T_2)}}\Prob{E_{S(T_2)}}.
\end{equation}

In order for $T_2$ to be observed by $TS4c_2$ it is required that wen the last edge of $T_2$ is observed on the stream at $t^M_{1,3,5}$ its two remaining edges are kept in $S_e$. 
Lemma~\ref{lem:reservoirhighorder} we have:
\begin{equation}\label{eq:t1seen2}
	\Prob{E_{S(T_2)}} = \frac{M_e}{t_{1,3,5}-1}\frac{M_e-1}{t_{1,3,5}-2},
\end{equation}
and:
\begin{equation}\label{eq:t2stored2}
	\Prob{E_{T_2}|E_{S(T_2)}}=\frac{M_e}{\tau^{t_6}}.
\end{equation}

Let us now consider  $\Prob{E_{T_1}|E_{T_2}}$. In order for $T_1$ to be found in $\Samp_\Delta$ at $t_6$ it is necessary for $T_1$ to have been observed by \algomultitwo{}. We thus have that $\Prob{E_{T_1}|E_{T_2}}=\Prob{E_{T_1}|E_{S(T_1)},E_{T_2}}\Prob{E_{S(T_1)}|E_{T_2}}$

\begin{equation}
\label{eq:t1stored2}
\Prob{E_{T_1}|E_{S(T_1)}}= \frac{M_\Delta - 1}{ \tau^{t_6} - 1}
\end{equation}

Let us now consider $\Prob{E_{S(T_1)}|E_{T_2}}$. while the content of $\Samp_\Delta$ itself does not influence the content of $\Samp_e$, the fact that $T_2$ is maintained in $\Samp_\Delta$ at $t_6$ implies that it has been observed on the stream at a previous time. We thus have $\Prob{E_{S(T_1)}|E_{T_2}} = \Prob{E_{S(T_1)}|E_{S(T_2)}}$. In order to study $p' = \Prob{E_{S(T_1)}|E_{S(T_2)}}$ it is necessary to distinguish the possible (5!) different arrival orders for edges $e_1$,$e_2$, $e_3$, $e_4$ and $e_5$. An efficient analysis we however reduce the number of cases to be considered to thirteen.
\begin{itemize}
\item $t^M_{1,2,4}\leq M_e$: in this case all edges of $T_1$ are observed on the stream before $M_e$ so they are deterministically inserted in $\Samp_e$ and thus $p'=1$.
\item $t_1>t^M_{2,3,4,5}$: in this case the edge $e_1$ that is shared by $T_1$ and $T_2$ is observed after $e_2$, $e_3$, $e_4$, $e_5$.
\begin{align*}
p' &= P(\{e_2,e_4\} \in {\Samp_e}^{(t_1)} | \{e_3,e_5\} \in {\Samp_e}^{(t_1)} ) \\
&= P(e_2 \in \Samp_e(t_1)| \{e_3,e_4, e_5\} \in {\Samp_e}^{(t_1)} ) \cdot P(e_4 \in {\Samp_e}^{(t_1)}| \{e_3,e_5\} \in {\Samp_e}^{(t_1)} ) \\
&= min(1, \frac{M - 3}{t_1 -4})\cdot min(1,\frac{M- 2}{t_1 - 3})
\end{align*}

\item $t^M_{3,5}>t_1>\max\{t^m_{3,5},t_2,t_4\}$: in this case only one of the edges $e_3$,$e_5$ is observed after $e_1$, which is observed after all the remaining edges. Here we consider the case when $e_3$ is the edge to be observed last. The same procedure follows for $e_5$ as well.
\begin{align*}
p' &= P(\{e_2,e_4\} \in {\Samp_e}^{(t_1)} | e_5 \in {\Samp_e}^{(t_1)} ) \\
&= P(e_2 \in \Samp_e(t_1)| \{e_4, e_5\} \in {\Samp_e}^{(t_1)} ) \cdot P(e_4 \in {\Samp_e}^{(t_1)}| e_5 \in {\Samp_e}^{(t_1)} ) \\
&= min(1, \frac{M - 2}{t_1 -3})\cdot min(1,\frac{M- 1}{t_1 - 2})
\end{align*}

\item $t^M_{3,5}>t^M_{2,4}>\max\{t^m_{3,5},t^m_{2,4},t_1\}$: in this case only one of the edges $e_3$, $e_5$ is observed after one of the edges $e_2$, $e_4$, which is observed after all the remaining edges. We consider the case when $e_3$ and $e_2$ are observed last. The same procedure follows for $e_4$ and $e_5$ as well.
\begin{align*}
p' &= P(\{e_1,e_4\} \in {\Samp_e}^{(t_2)} | \{e_1,e_5\} \in {\Samp_e}^{(t_3)} ) \\
&= P(e_1 \in \Samp_e(t_2)| \{e_1,e_4, e_5\} \in {\Samp_e}^{(t_3)} ) \cdot P(e_4 \in {\Samp_e}^{(t_2)}| \{e_1,e_5\} \in {\Samp_e}^{(t_3)} ) \\
&= 1\cdot min(1,\frac{M- 2}{t_2 - 3})
\end{align*}

\item $t^m_{3,5}>t_1>t^M_{2,4}$:  in this case both  edges $e_3$,$e_5$ are observed after $e_1$, which is observed after all the remaining edges. Here we consider the case when $t_3 > t_5 > t_1$. The same procedure follows for $t_5 > t_3 > t_1$.
\begin{align*}
p' &= P(\{e_2,e_4\} \in {\Samp_e}^{(t_1)} | \{e_1, e_5\} \in {\Samp_e}^{(t_3)} ) \\
&= P(e_2 \in \Samp_e(t_1)| e_4 \in \Samp_e(t_1), \{e_1, e_5\} \in {\Samp_e}^{(t_3)} ) \cdot P(e_4 \in {\Samp_e}^{(t_1)}|  \{e_1, e_5\} \in {\Samp_e}^{(t_3)} ) \\
&= min(1, \frac{M - 1}{t_1 -2})\cdot min(1,\frac{M}{t_1 - 1})
\end{align*}

\item $t^m_{3,5}>t^M_{2,4}>t_1$: in this case both edges $e_3$,$e_5$ are observed after one of the edges $e_2$,$e_4$, which is observed after $e_1$. We consider the case $t_2 > t_4$ and $t_3 > t_5$. The same procedure follows for $t_4 > t_2$ and $t_5 > t_3$
\begin{align*}
p' &= P(\{e_1,e_4\} \in {\Samp_e}^{(t_2)} | \{e_1, e_5\} \in {\Samp_e}^{(t_3)} ) \\
&= P(e_4 \in \Samp_e(t_2)| e_1 \in \Samp_e(t_2), \{e_1, e_5\} \in {\Samp_e}^{(t_3)} ) \cdot P(e_1 \in {\Samp_e}^{(t_3)}|  \{e_1, e_5\} \in {\Samp_e}^{(t_3)} ) \\
&= min(1, \frac{M - 1}{t_2 -2})\cdot 1
\end{align*}

\item  $t^M_{2,4}>t_1>\max\{M_e,t^m_{2,4},t^M_{3,5}\}$: in this case both edges $e_2$, $e_4$ are observed after $e_1$, which is observed after the edge reservoir is filled and after all the remaining edges. We consider the case when $t_2>t_4$. The same procedure follows for $t_4>t_2$. 
\begin{align*}
p' &= P(\{e_1,e_4\} \in {\Samp_e}^{(t_2)} | \{e_3, e_5\} \in {\Samp_e}^{(t_1)} ) \\
&= P(e_4 \in \Samp_e(t_2)| e_1 \in \Samp_e(t_2), \{e_3, e_5\} \in {\Samp_e}^{(t_3)} ) \cdot P(e_1 \in {\Samp_e}^{(t_2)}|  \{e_3, e_5\} \in {\Samp_e}^{(t_1)} ) \\
&= \frac{M - 1}{t_2 -2} \cdot  \frac{M - 2}{t_1 -3} \cdot \frac{t_1 - 1}{t_2 -1}
\end{align*}

\item $t^M_{2,4}>M_e>t_1>\max\{t^m_{2,4},t^M_{3,5}\}$:in this case both edges $e_2$, $e_4$ are observed after $e_1$, which is observed before the edge reservoir is filled and after all the remaining edges. We consider the case when $t_2>t_4$. The same procedure follows for $t_4>t_2$. 
\begin{align*}
p' &= P(\{e_1,e_4\} \in {\Samp_e}^{(t_2)} | \{e_3, e_5\} \in {\Samp_e}^{(t_1)} ) \\
&= P(e_4 \in \Samp_e(t_2)| e_1 \in \Samp_e(t_2), \{e_3, e_5\} \in {\Samp_e}^{(t_3)} ) \cdot P(e_1 \in {\Samp_e}^{(t_2)}|  \{e_3, e_5\} \in {\Samp_e}^{(t_1)} ) \\
&= \frac{M - 1}{t_2 -2} \cdot  \frac{M_e}{t_2 -1}
\end{align*}

\item $t^M_{2,4}>t^M_{3,5}>\max\{ M_e,t^m_{3,5},t^m_{2,4},t_1\}$: in this case only one of the edges $e_2$, $e_4$ is observed after one of the edges $e_3$, $e_5$ which is observed after the edge reservoir if filled and after all the remaining edges. We consider the case when $t_2>t_4$ and $t_3>t_5$. The same procedure follows for $t_4>t_2$ and $t_5>t_3$
\begin{align*}
p' &= P(\{e_1,e_4\} \in {\Samp_e}^{(t_2)} | \{e_1, e_5\} \in {\Samp_e}^{(t_3)} ) \\
&= P(e_1 \in \Samp_e(t_2)| e_4 \in \Samp_e(t_2), \{e_1, e_5\} \in {\Samp_e}^{(t_3)} ) \cdot P(e_4 \in {\Samp_e}^{(t_2)}|  \{e_1, e_5\} \in {\Samp_e}^{(t_3)} ) \\
&= \frac{t_3 - 1}{t_2 -1} \cdot  \frac{t_3 - 2}{t_2 -2} \cdot \frac{M_e - 2}{t_3 -3}
\end{align*}

\item $t^M_{2,4}>M_e>t^M_{3,5}>\max\{t^m_{3,5},t^m_{2,4},t_1\}$: in this case only one of the edges $e_2$, $e_4$ is observed after one of the edges $e_3$, $e_5$ which is observed before the edge reservoir if filled and after all the remaining edges. We consider the case when $t_2>t_4$ and $t_3>t_5$. The same procedure follows for $t_4>t_2$ and $t_5>t_3$
\begin{align*}
p' &= P(\{e_1,e_4\} \in {\Samp_e}^{(t_2)} | \{e_1, e_5\} \in {\Samp_e}^{(t_3)} ) \\
&= P(e_1 \in \Samp_e(t_2)| e_4 \in \Samp_e(t_2), \{e_1, e_5\} \in {\Samp_e}^{(t_3)} ) \cdot P(e_4 \in {\Samp_e}^{(t_2)}|  \{e_1, e_5\} \in {\Samp_e}^{(t_3)} ) \\
&= \frac{M_e}{t_2 -1} \cdot  \frac{M_e - 1}{t_2 -2}
\end{align*}
\item $t^m_{2,4}>t_1>t^M_{3,5}$:  in this case both  edges $e_2$,$e_4$ are observed after $e_1$, which is observed after all the remaining edges. Here we consider the case when $t_2 > t_4 > t_1$. The same procedure follows for $t_4 > t_2 > t_1$.
\begin{align*}
p' &= P(\{e_2,e_4\} \in {\Samp_e}^{(t_1)} | \{e_3, e_5\} \in {\Samp_e}^{(t_1)} ) \\
&= P(e_2 \in \Samp_e(t_1)| e_4 \in \Samp_e(t_1), \{e_3, e_5\} \in {\Samp_e}^{(t_1)} ) \cdot P(e_4 \in {\Samp_e}^{(t_1)}|  \{e_3, e_5\} \in {\Samp_e}^{(t_1)} ) \\
&= min(1, \frac{M_e - 1}{t_2 -2})\cdot min(1,\frac{M_e}{t_2 - 1})
\end{align*}

\item $t^m_{2,4}>t^M_{3,5}>\max\{M,t_1\}$: in this case both edges $e_2$,$e_4$ are observed after one of the edges $e_3$,$e_5$, which is observed after the edge reservoir is filled and after all the remaining edges.
\begin{align*}
p' &= P(\{e_1,e_4\} \in {\Samp_e}^{(t_2)} | \{e_1, e_5\} \in {\Samp_e}^{(t_3)} ) \\
&= P(e_4 \in \Samp_e(t_2)| e_1 \in \Samp_e(t_2), \{e_1, e_5\} \in {\Samp_e}^{(t_3)} ) \cdot P(e_1 \in {\Samp_e}^{(t_2)}|  \{e_1, e_5\} \in {\Samp_e}^{(t_3)} ) \\
&= \frac{M_e - 1}{t_2 -2} \cdot \frac{t_3 - 1}{t_2 - 1}
\end{align*}

\item $t^m_{2,4}>M_e>t^M_{3,5}>t_1$: in this case both edges $e_2$,$e_4$ are observed after one of the edges $e_3$,$e_5$, which is observed before the edge reservoir is filled and after all the remaining edges.
\begin{align*}
p' &= P(\{e_1,e_4\} \in {\Samp_e}^{(t_2)} | \{e_1, e_5\} \in {\Samp_e}^{(t_3)} ) \\
&= P(e_4 \in \Samp_e(t_2)| e_1 \in \Samp_e(t_2), \{e_1, e_5\} \in {\Samp_e}^{(t_3)} ) \cdot P(e_1 \in {\Samp_e}^{(t_2)}|  \{e_1, e_5\} \in {\Samp_e}^{(t_3)} ) \\
&= \frac{M_e - 1}{t_2 -2} \cdot \frac{M_e}{t_2 - 1}
\end{align*}

\end{itemize}

The lemma follows combining the result for the values of $p'$ with (\ref{eq:t1seen2}), (\ref{eq:t2stored2}) and (\ref{eq:t1stored2}) in (\ref{eq:prob2}).
\end{proof}

We now present the \emph{unbiasedness}  of the estimations obtained using \algomultitwo{}.

\begin{proof}[Proof of Theorem~\ref{thm:multiunbiased2}]
Let $t^*$ denote the first step at for which the number of triangle seen exceeds $M_\Delta$. 
For $t \leq \min \{M_e, t^*\}$, the entire graph $G^{(t)}$ is maintained in $\Samp_e$ and all the triangles in $G^{(t)}$ are stored in $\Samp_\Delta$. Hence all the cliques in $G^{(t)}$ are deterministically observed by \algomulti and we have $\varkappa^{(t)}= |\clisett{4}{t}|$.

  Assume now $t>\min \{M_e, t^*\}$ and assume that $|\clisett{4}{t}|>0$, otherwise, the algorithm deterministically returns $0$ as an estimation and the thesis follows. 
  For any 4-clique $\lambda\in\clisett{4}{t}$ which is observed by \algomultitwo{} with probability $p_\lambda$, consider a random variable $X_\lambda$ which takes value $p^{-1}$ iff $\lambda$ is acutally observed by \algomultitwo{} (i.e., with probability $p_\lambda$) or zero otherwise. We thus have $\mathbb{E}\left[X_\lambda\right] = p_\lambda^{-1}\Prob{X_\lambda = p_\lambda^{-1}} = p_\lambda^{-1}p_\lambda = 1$.    Recall that every time \algomultitwo{} observes a 4-clique on the stream it evaluates the  probability $p$ (according its correct value as shown in Lemma~\ref{lem:probclique2}) of observing it and it correspondingly increases the running estimator by $p^{-1}$.  We therefore can express the running estimator $\varkappa^{(t)}$ as:
  \[
    \varkappa^{(t)}=\sum_{\lambda\in\clisett{4}{t}}X_\lambda\enspace.
  \]
  From linearity of expectation, we thus have:
  \[
    \mathbb{E}\left[\varkappa^{(t)}\right]=\sum_{\lambda\in\clisett{4}{t}}
    \mathbb{E}[X_\lambda]
    =\sum_{\lambda\in\clisett{4}{t}}
    p_\lambda^{-1} p_\lambda = |\clisett{4}{t}|.
  \]
\end{proof}

\newpage
\section{Additional theoretical results for \algosingle{}}

\begin{lemma}\label{lem:singlecount}
	Let $\lambda\in \clisett{4}{t}$ with $\lambda = \{e_1,e_2,e_3,e_4,e_5,e_6\}$. Assume, without loss of generality, that the edge $e_i$ is observed at $t_i$ (not necessarily consecutively) and that $t_6>\max\{t_i,1\leq i\leq5\}$. $\lambda$ is observed by \algosingle{} at time $t_6$ with probability:
	\begin{equation}
		p_\lambda = 
		\begin{cases}
			0&\ if |M|<5,\\
			1&\ if\ t_6\leq M+1,\\
			\prod_{i=0}^{5}\frac{M-i}{t-i-1}&\ if\ t_6>M+1.	
 		\end{cases}
	\end{equation}
\end{lemma}
\begin{proof}
	Clearly \algosingle{} can observe $\lambda$ \emph{only} at the time step at which the last edge $e_6$ is observed on the stream at $t_6$. Further, from its construction, \algosingle{} observed a 4-clique $\lambda$ if and only if when its large edge is observed on the stream its remaining five edges are kept in the edge reservoir $\Samp$. $\Samp$ is an uniform edge sample maintained by means of the reservoir sampling scheme. From Lemma~\ref{lem:reservoirhighorder} we have that the probability of any five elements observed on the stream \emph{prior} to $t_6$ being in $\Samp$ at the beginning of step $t_6$ is given by:
	\begin{equation*}
		\Prob{\{e_1,e_2,e_3,e_4,e_5,e_6\}\subseteq } = 
		\begin{cases}
			0&\ if |M|<5,\\
			1&\ if\ t_6\leq M+1,\\
			\prod_{i=0}^{4}\frac{M-i}{t-i-1}&\ if\ t_6>M+1.	
 		\end{cases}	
	\end{equation*}
	The lemma follows.
\end{proof}

We now present the proof of the \emph{unbiasedness} of the estimations obtained using \algosingle{}.

\begin{proof}[Proof of Theorem~\ref{thm:singleunbiased}]
Recall that when a new edge $e_t$ is observed on the stream \algosingle{} updates the estimator $\varkappa^{(t)}$ \emph{before} deciding whether the new edge is inserted in $\Samp$. For $t \leq M+1$, the entire graph $G^{(t)}\setminus\{e_t\}$ is maintained in $\Samp$, thus whenever an edge $e_t$ is inserted at time $t \leq M+1$, \algosingle{} observes \emph{all} the triangles which include $e_t$ in $G^{(t)}$ with probability 1 thus increasing $\varkappa$ by one. By a simple inductive analysis we can therefore conclude that for $t \leq M+1$ we have $\varkappa^{(t)}=|\clisett{4}{t}|$. 

  Assume now $t>M+1$ and assume that $|\clisett{4}{t}|>0$, otherwise, the algorithm deterministically returns $0$ as an estimation and the thesis follows. 
  Recall that every time \algosingle{} observes a 4-clique on the stream  a time $t$ it computes the probability $p = \prod_{i=0}^{4}\frac{M-i}{t-i-1}$ of observing it and it correspondingly increases the running estimator by $p^{-1}$. 
From Lemma~\ref{lem:reservoirhighorder}, the probability $p$ computed by \algosingle{} does indeed correspond to the correct probability of observing a 4-clique at time $t$. 
  For any 4-clique $\lambda\in\clisett{4}{t}$ which is observed by \algomulti with probability $p_\lambda$, consider a random variable $X_\lambda$ which takes value $p^{-1}$ iff $\lambda$ is actually observed by \algosingle{} (i.e., with probability $p_\lambda$) or zero otherwise. We thus have $\mathbb{E}\left[X_\lambda\right] = p_\lambda^{-1}\Prob{X_\lambda = p_\lambda^{-1}} = p_\lambda^{-1}p_\lambda = 1$.   We therefore can express the running estimator $\varkappa^{(t)}$ as:
  \[
    \varkappa^{(t)}=\sum_{\lambda\in\clisett{4}{t}}X_\lambda\enspace.
  \]
  From linearity of expectation, we thus have:
  \[
    \mathbb{E}\left[\varkappa^{(t)}\right]=\sum_{\lambda\in\clisett{4}{t}}
    \mathbb{E}[X_\lambda]
    =\sum_{\lambda\in\clisett{4}{t}}
    p_\lambda^{-1} p_\lambda = |\clisett{4}{t}|.
  \]
\end{proof}

We now present the proof for the upper bound for the variance of the estimations obtained using \algosingle{}.

\begin{proof}[Proof of Theorem~\ref{thm:variancesingle}]
Assume $|\clisett{4}{t}|>0$ and $t>M+1$, otherwise (from Theorhem~\ref{thm:singleunbiased}) \algomultione{} estimation is
  deterministically correct and has variance 0 and the thesis holds. For each $\lambda\in\clisett{4}{t}$ let $\lambda = \{e_1,e_2,e_3,e_4,e_5,e_6\}$, without loss of generality let us assume the edges are disposed as in Figure~\ref{fig:4clidetection}.  Assume further, without loss of generality, that the edge $e_i$ is observed at $t_i$ (not necessarily consecutively) and that $t_6>\max\{t_i,1\leq i\leq5\}$. Let us consider the random variable $\delta_{\lambda}$ (which takes value $p^{-1}_{\lambda}$ if the 4-clique $\lambda$ is observed by \algosingle{}, or zero otherwise. 
  From Lemma~\ref{lem:singlecount}, we have:
  \begin{equation*}
  	p_{\lambda} = \Prob{\delta_{\lambda} = p^{-1}_\lambda} = \prod_{i=0}^{4}\frac{M-i}{t-i-1},
  \end{equation*}
  and thus:
  \begin{equation*}
  	 \Var{\delta_{\lambda}} = p_{\lambda}^{-1}-1.
  \end{equation*}  
  We can express the estimator $\varkappa^{(t)}$ as $\varkappa^{(t)}=\sum_{\lambda\in \clisett{4}{t}}$. We therefore have:
  \begin{align}\label{eq:variasing}
    \Var{\varkappa^{(t)}}
     &=\Var{\sum_{\lambda\in\clisett{4}{t}} \delta_{\lambda}}\nonumber\\\
     &=\sum_{\lambda\in \clisett{4}{t}} \sum_{\gamma \in
    \clisett{4}{t}}\text{Cov}\left[\delta_{\lambda},\delta_{\gamma}\right]\nonumber\\
     &= \sum_{\lambda\in \clisett{4}{t}}\Var{\delta_{\lambda}}+\sum_{\substack{\lambda,
    \gamma\in \clisett{4}{t}\\\lambda \neq \gamma}} \Cov{\delta_{\lambda},\delta_{\gamma}}\nonumber\\
     &\leq |\clisett{4}{t}|\left(\prod_{i=0}^{4}\frac{t-1-i}{M-i}-1\right)+\sum_{\substack{\lambda,
    \gamma\in \clisett{4}{t}\\\lambda \neq \gamma}} \Cov{\delta_{\lambda},\delta_{\gamma}} \enspace.
  \end{align}
  
 From Lemma~\ref{lem:sharingclique}, we have that two distinct cliques $\lambda$ and $\gamma$ can share one, three or no edges. In analyzing the summation of covariance terms appearing in the right-hand-side of \eqref{eq:variacemul1n1} we shall therefore consider separately the pairs that share respectively one, three or no edges.
\begin{itemize}
   	\item $\lambda$ and $\gamma$ do not share any edge:
   	\begin{align*}
    	\Exp{\delta_{\lambda}\delta_{\gamma}} &= p_{\lambda}^{-1}p_{\gamma}^{-1} \Prob{\delta_{\lambda}=p^{-1}_{\lambda} \wedge \delta_{\gamma}=p^{-1}_{\gamma}} \\
    	&= p_{\lambda}^{-1}p_{\gamma}^{-1} \Prob{\delta_{\lambda}=p^{-1}_{\lambda}|\delta_{\gamma}=p^{-1}_{\gamma}}\Prob{\delta_{\lambda}=p^{-1}_{\lambda}}\\
    \end{align*}
    The term $\Prob{\delta_{\lambda}=p^{-1}_{\lambda}|\delta_{\gamma}=p^{-1}_{\gamma}}$ denotes the probability of \algosingle{} observing $\lambda$ conditioned of the fact that $\gamma$ was observed. Note that as $\lambda$ and $\gamma$ do not share any edge, no edge of $\gamma$ will be used by \algosingle{} to detect $\lambda$. Rather, if any edge of $\gamma$ is included in $\Samp$, this lowers the probability of \algosingle{} detecting $\lambda$ as some of space available in $\Samp$ may be occupied by edges of $\gamma$. Therefore we have $ \Prob{\delta_{\lambda}=p^{-1}_{\lambda}|\delta_{\gamma}=p^{-1}_{\gamma}}\leq  \Prob{\delta_{\lambda}=p^{-1}_{\lambda}}$ and thus:
    \begin{align*}
    	\Exp{\delta_{\lambda}\delta_{\gamma}} &= p_{\lambda}^{-1}p_{\gamma}^{-1} \Prob{\delta_{\lambda}=p^{-1}_{\lambda}|\delta_{\gamma}=p^{-1}_{\gamma}}\Prob{\delta_{\gamma}=p^{-1}_{\gamma}}\\
    	&\leq p_{\lambda}^{-1}p_{\gamma}^{-1} \Prob{\delta_{\lambda}=p^{-1}_{\lambda}}\Prob{\delta_{\gamma}=p^{-1}_{\gamma}}\\
    	&\leq p_{\lambda}^{-1}p_{\gamma}^{-1} p_{\lambda}p_{\gamma}\\
    	&\leq 1.
    \end{align*}
    As $\Cov{\delta_{\lambda},\delta_{\gamma}} = \Exp{\delta_{\lambda}\delta_{\gamma}} -1$, we therefore have $\Cov{\delta_{\lambda},\delta_{\gamma}}\leq 0$. Hence we can conclude that the contribution of the covariances of the pairs of random variables corresponding to 4-cliques that do not share any edge to the summation in~\eqref{eq:variacemul1n1} is less or equal to zero.
    \item $\lambda$ and $\gamma$ share exactly one edge $e^*=\lambda \cap \gamma$ (as shown in Figure~\ref{fig:share1edge}).  Let us consider the event $E^*=$``$e^*\in \Samp_{t_6}$ unless $e^*$ is observed at $t_6$''. Clearly $\Prob{\delta_{\lambda}=p^{-1}_{\lambda}|\delta_{\gamma}=p^{-1}_{\gamma}}\leq \Prob{\delta_{\lambda}=p^{-1}_{\lambda}|E*}$. Recall from Lemma~\ref{lem:singlecount} that $\Prob{\delta_{\lambda}|E^*} = \Prob{\{e_1,\ldots,e_5\}\subseteq \Samp_e^{(t_6)}|E^*}$. We can distinguish two cases: (a) $e^*=e_6$: in this case we have $\Prob{\delta_{\lambda}|E^*}=\Prob{\{e_1,\ldots,e_5\}\subseteq \Samp_e^{(t_6)}} = p_{\lambda}$; (b) $e^*\neq e_6$: in this case we have $\Prob{\delta_\lambda|E^*}=\Prob{ \{e_1,\ldots,e_5\} \setminus \{e^*\}\subseteq \Samp^{(t_6)}|e^* \in \Samp_e^{(t_6)}} = \prod_{i=0}^{3}\frac{M-1-i}{t_6-2-i}$. We can therefore conclude:
    \begin{align*}
    	\Exp{\delta_{\lambda}\delta_{\gamma}} &= p_{\lambda}^{-1}p_{\gamma}^{-1} \Prob{\delta_{\lambda}=p^{-1}_{\lambda}|\delta_{\gamma}=p^{-1}_{\gamma}}\Prob{\delta_{\gamma}=p^{-1}_{\gamma}}\\
    	&\leq p_{\lambda}^{-1}\Prob{\delta_{\lambda}=p^{-1}_{\lambda}|\delta_{\gamma}=p^{-1}_{\gamma}}\\
    	&\leq \frac{t_6-1}{M}.
    \end{align*}
    We can thus bound the contribution to covariance summation in ~\eqref{eq:variasing} given by the pairs of random variables corresponding to 4-cliques that share one edge as:
    \begin{equation}\label{eq:varsingshareonedge}
    	a^{(t)}\left(\frac{t-1}{M}-1\right),
    \end{equation}
    where $a^{(t)}$ denotes the number of unordered pairs of 4-cliques which share one edge in $G^{(t)}$.
   \item $\lambda$ and $\gamma$ share three edges $\{e^*_1,e^*_2,e^*_3\}$ which form a triangle sub-structure for both $\lambda$ and $\gamma$. Let us refer to Figure~\ref{fig:share3edges}. Let us consider the event $E^*=$``$\{e^*_1,e^*_2,e^*_3\}\cap E^{(t_6-1)}\subseteq \Samp^{(t_6)}$. Clearly $\Prob{\delta_{\lambda}=p^{-1}_{\lambda}|\delta_{\gamma}=p^{-1}_{\gamma}}\leq \Prob{\delta_{\lambda}=p^{-1}_{\lambda}|E*}$. Recall  that $\Prob{\delta_{\lambda}|E^*} = \Prob{\{e_1,\ldots,e_5\} \subseteq \Samp^{(t_6)}|E^*}$. We can distinguish two cases: (a) $e_6\cap \{e^*_1,e^*_2,e^*_3\} \neq \emptyset$: in this case we have $|\{e_1,\ldots, e_5\} \setminus\left(\{e^*_1,e^*_2,e^*_3\}\setminus\{e_6\}\right)|= 3$ hence $\Prob{\delta_{\lambda}|E^*}=\Prob{\{e_1,\ldots, e_5\} \setminus\left(\{e^*_1,e^*_2,e^*_3\}\setminus\{e_6\}\right)|\left(\{e^*_1,e^*_2,e^*_3\}\setminus\{e_6\}\right)\subseteq \Samp^{(t_6)}} = \prod_{i=0}^{2}\frac{M-2-i}{t_6-3-i}$; (b) $e_6\cap \{e^*_1,e^*_2,e^*_3\} = \emptyset$: in this case we have $|\{e_1,\ldots, e_5\} \setminus\left(\{e^*_1,e^*_2,e^*_3\}\setminus\{e_6\}\right)|= 2$ hence $\Prob{\delta_{\lambda}|E^*}=\Prob{\{e_1,\ldots, e_5\} \setminus \{e^*_1,e^*_2,e^*_3\}|\{e^*_1,e^*_2,e^*_3\})\subseteq \Samp^{(t_6)}} = \prod_{i=0}^{1}\frac{M-3-i}{t_6-4-i}$;.
   For the pairs of 4-cliques which share three edge we therefore have:
    \begin{align*}
    	\Exp{\delta_{\lambda}\delta_{\gamma}} &= p_{\lambda}^{-1}p_{\gamma}^{-1} \Prob{\delta_{\lambda}=p^{-1}_{\lambda}|\delta_{\gamma}=p^{-1}_{\gamma}}\Prob{\delta_{\gamma}=p^{-1}_{\gamma}}\\
    	&\leq p_{\lambda}^{-1}\Prob{\delta_{\lambda}=p^{-1}_{\lambda}|\delta_{\gamma}=p^{-1}_{\gamma}}\\
    	&\leq \prod_{i=0}^{2}\frac{t-1-i}{M-i}.
    \end{align*}
	We can thus bound the contribution to covariance summation in ~\eqref{eq:variasing} given by the pairs of random variables corresponding to 4-cliques that share one edge as:
    \begin{equation}\label{eq:varsingsharethreedge}
    	b^{(t)}\left(\prod_{i=0}^{2}\frac{t-1-i}{M-i}-1\right)= b^{(t)}c'\left(\frac{t}{M}\right)^5
    \end{equation}
      where $b^{(t)}$ denotes the number of unordered pairs of 4-cliques which share three edges in $G^{(t)}$ and $c'=(M-1)(M-2)/M^2$.
    \end{itemize} 
The Theorem follows form the previous considerations and by combining by combining~\eqref{eq:varsingshareonedge} and~\eqref{eq:varsingsharethreedge} in~\eqref{eq:variasing}.
\end{proof}

\section{Pseudocode for \algodynamic}

\subsection{Algorithm \algomulti{} description}
\begin{algorithm}[H]
  \caption{\algodynamic{} - Adaptive Version of \tiesa{}}
  \label{alg:ncliquecounter}
  \begin{algorithmic}[1]
  \Statex{\textbf{Input:} Insertion-only edge stream $\Sigma$, integers $M$}
    \State $\Samp_e \leftarrow\emptyset$ , $\Samp_\Delta \leftarrow\emptyset$, $\Samp_\Delta' \leftarrow\emptyset$, $M_\Delta' \leftarrow 0$, $t\leftarrow 0$, $t_\Delta\leftarrow 0$, $\sigma\leftarrow 0$, $r\leftarrow 1$
    \For{ {\bf each} element $(u,v)$ from $\Sigma$}
    \State $t\leftarrow t +1$
    \If {$r = 1$} \Comment{If we are in first regimen}
    		\State $\alpha = \textsc{SWITCH}()$
       		 \If {$\alpha = 0$} \Comment{Remain in the first regimen}
       			 \State \textsc{Update4CliquesFirstRegimen}$(u,v)$
     		 \Else  \Comment{First time switching}
		 	\State $r \leftarrow 2$
		 	\State $\textsc{CreateTriangleReservoir}(\alpha)$
			\State \textsc{Update4CliquesSecondRegimen}$(u,v)$
      		 \EndIf
      \ElsIf{$r = 2$}  \Comment{We have already switched to second regimen}
       		\If{$t\%M=0$}
      			\State $\textsc{UpdateMemory}()$
   		\EndIf
		\If{$t_\Delta' < M_\Delta$}
      			\State $\textsc{UpdateTriangles}((u,v),t_\Delta', S_\Delta')$
   		\EndIf
		\State $p_\Delta \leftarrow min(1, \frac{M_\Delta}{t_\Delta})$,  $p_\Delta' \leftarrow min(1, \frac{M_\Delta'}{t_\Delta'})$
		\If{$p_\Delta < p_\Delta'$}
      			\State $S_{merged} \leftarrow \emptyset$
			\For{\textbf{each} $(x,y,z) \in S_\Delta$}
			\If {\textsc{FlipBiasedCoin}$(\frac{p_\Delta'}{p_\Delta}) = $ heads}
				\State $S_{merged} \leftarrow S_{merged} \cup {(x,y,z)} $
			\EndIf
			\EndFor
			\State $S_\Delta =  S_{merged} \cup S_\Delta'$
			
   		\EndIf
		\State \textsc{Update4CliquesSecondRegimen}$(u,v)$
		\State $\textsc{UpdateTriangles}((u,v),S_\Delta)$
		
       \EndIf

    \If{\textsc{SampleEdge}$((u,v), t )$}
      \State $\Samp \leftarrow \Samp\cup \{((u,v), t)\}$
    \EndIf
    \EndFor
    
    \Statex
    \Function{\textsc{Switch}}$$
    	 \State $p_s \leftarrow min(1,\frac{M}{t})^5$ 
	 \State $p_\alpha \leftarrow min(1, \alpha\frac{M}{t})^4 \cdot min(1,(1-\alpha)\frac{M}{t_\Delta})^2$ \textbf{where} $a=argmax_{\alpha\in[\frac{2}{3},1]} p_\alpha$
	 \If{$p_s < p_\alpha$}
      		\State \textbf{return} $\alpha$
	\Else
		\State \textbf{return} $0$
     \EndIf   

    \EndFunction
    
    \Statex
    \Function{\textsc{CreateTriangleReservoir}}{$\alpha$}
    \State $i \leftarrow 1$
			\State $M_\Delta = (1 - \alpha)M$
        			\While{$|S_\Delta| < M_\Delta$}
				\State $(x,y)^{(i)} \leftarrow Edge\ observed\ at\ time\ i$
				 \For{ {\bf each} element $z$ from $\mathcal{N}^{\Samp_{x,y}^{(i)}}$}
				  \State $t_\Delta \leftarrow t_\Delta+1$ 
       				\State $S_\Delta \leftarrow S_\Delta \cup {(x,y,z)}$

				 \EndFor
			\EndWhile
			\While{$|S| > M - M_\Delta$}
				\State $(v,w) \leftarrow random\ edge\ from\ S$
				  \State $\Samp\leftarrow \Samp\setminus \{(v,w)\}$
			\EndWhile
    \EndFunction
\algstore{myalg}
\end{algorithmic}
\end{algorithm}
\begin{algorithm}[H]                   
\begin{algorithmic} [1]                   
\algrestore{myalg}
\Function{UpdateTriangles}{$(u,v), t $}
      \State $\mathcal{N}^\Samp_{u,v} \leftarrow \mathcal{N}^\Samp_u \cap \mathcal{N}^\Samp_v$ 
      \For{ {\bf each} element $w$ from $\mathcal{N}^\Samp_{u,v}$}
       	\State $t_\Delta \leftarrow t_\Delta + 1$
      	\If{\textsc{SampleTriangle}$(u,v,w)$}
      		\State $\Samp_\Delta \leftarrow \Samp_\Delta \cup \{u,v,w\}$
      	\EndIf
      \EndFor
    \EndFunction   
 \Statex
    \Function{\textsc{SampleTriangle}}{$u,v,w$} 
         \If {$t_\Delta\leq  M_\Delta$}
        \State \textbf{return} True
      \ElsIf{\textsc{FlipBiasedCoin}$(\frac{M_\Delta}{t_\Delta}) = $ heads}
        \State $(u_1,v_1,w_1) \leftarrow$ random triangle from $\Samp_\Delta$
        \State $\Samp_\Delta \leftarrow \Samp_\Delta\setminus \{(u_1,v_1,w_1)\}$
        \State \textbf{return} True
      \EndIf
      \State \textbf{return} False
    \EndFunction

       \Statex
    \Function{\textsc{UpdateMemory}}$$
    	\State $t^{(i+1)M}_\Delta = 2t_\Delta^{(iM)} - t_\Delta^{((i-1)M)}$ \Comment{Prediction of number of triangles at (i+1) step}
    	 \State $p_s \leftarrow min(1,\frac{M}{t})^5$ 
	 \State $p_\alpha \leftarrow min(1, \alpha\frac{M}{t})^4 \cdot min(1,(1-\alpha)\frac{M}{t^{(i+1)M}_\Delta})^2$ \textbf{where} $a=argmax_{\alpha\in[\frac{2}{3},1]} p_\alpha$
	 \If{$\alpha < \frac{M_\Delta}{M}$}
      		\For{$i \in [1,\frac{M_\Delta}{M} - \alpha]$}
				\State $(v,w) \leftarrow random\ edge\ from\ S$
				  \State $\Samp\leftarrow \Samp\setminus \{(v,w)\}$
				  \State $M_\Delta \leftarrow M_\Delta + 1$
				  \State $M_\Delta' \leftarrow M_\Delta' + 1$
				  
		\EndFor
     \EndIf   

    \EndFunction

\Statex
    \Function{Update4CliquesFirstRegimen}{$u,v$}
      \State $\mathcal{N}^\Samp_{u,v} \leftarrow \mathcal{N}^\Samp_u \cap \mathcal{N}^\Samp_v$
      \For{ {\bf each} element $(x,w)$ from $\mathcal{N}^\Samp_{u,v} \times \mathcal{N}^\Samp_{u,v}$}
       \If{$(x,w)$ in  $\Samp_e$} 
       		\If{$t\leq M$}
       			\State{$p\leftarrow 1$}
       		\Else
      			\State $p \leftarrow  min \{1, \frac{M(M-1)(M-2)(M-3)(M-4)}{(t-1)(t-2)(t-3)(t-4)(t-5)}\}$
      		\EndIf
      	\State $\varkappa \leftarrow \varkappa + p^{-1}$
    \EndIf
	\EndFor
   \EndFunction

    \Function{Update4CliquesSecondRegimen}{$(u,v), t$}
      \For{{\bf each} triangle $(u, w, z)\in \Samp_\Delta$}
        \If{$(v,w) \in \Samp_e \wedge (v,w) \in \Samp_e$}
      		\State $p \leftarrow \textsc{ProbClique}((u,w,z),(v,w), (v,z)) $
      		\State $\sigma \leftarrow \sigma + p^{-1}/2$
     \EndIf   
      \EndFor
      \For{{\bf each} triangle $(v, w, z)\in \Samp_\Delta$}
        \If{$(u,w) \in \Samp_e \wedge (u,z) \in \Samp_e$}
      		\State $p \leftarrow \textsc{ProbClique}((v,w,z),(u,w), (u,z)) $
      		\State $\sigma \leftarrow \sigma + p^{-1}/2$
     \EndIf   
      \EndFor
    \EndFunction

    \Statex
    \Function{\textsc{SampleEdge}}{$(u,v),t$} 
      \If {$t\leq  M$}
        \State \textbf{return} True
      \ElsIf{\textsc{FlipBiasedCoin}$(\frac{M}{t}) = $ heads}
        \State $((u',v'),t') \leftarrow$ random edge from $\Samp$
        \State $\Samp\leftarrow \Samp\setminus \{((u',v'),t')\}$
        \State \textbf{return} True
      \EndIf
      \State \textbf{return} False
    \EndFunction
  \end{algorithmic}
\end{algorithm}

\newpage
\section{Additional results on 5 clique counting}

\vskip -.15in 
\begin{algorithm}[H]
 \small 
  \caption{\algomultifive{} - Tiered Sampling for 5-Clique counting}
  \label{alg:ncliquecounterfive}
  \begin{algorithmic}[0]
  \Statex{\textbf{Input:} Insertion-only edge stream $\Sigma$, integers $M$, $M_C$}
    \State $\Samp_e \leftarrow\emptyset$ , $\Samp_C \leftarrow\emptyset$, $t\leftarrow 0$, $t_C\leftarrow 0$, $\varkappa\leftarrow 0$
    \For{ {\bf each} element $(u,v)$ from $\Sigma$}
    \State $t\leftarrow t +1$
    \State \textsc{Update5Cliques}$(u,v)$
    \State \textsc{Update4Cliques}$(u,v)$
    \If{\textsc{SampleEdge}$((u,v), t )$}
      \State $\Samp \leftarrow \Samp\cup \{((u,v), t)\}$
    \EndIf
    \EndFor
    \Statex
    \Function{Update5Cliques}{$(u,v), t$}
      \For{{\bf each} 4-clique $(u, x, w, z, t^*)\in \Samp_\Delta$}
        \If{$(v,x) \in \Samp_e \wedge (v,w) \in \Samp_e  \wedge (v,z) \in \Samp_e $}
      		\State $p \leftarrow  \frac{M_C}{t_C} \cdot (\frac{M_e}{t})^3  \cdot (\frac{M_e}{t^*})^5$
      		\State $\sigma \leftarrow \varkappa + p^{-1}/2$
     \EndIf   
      \EndFor
     \For{{\bf each} 4-clique $(v, x, w, z, t^*)\in \Samp_\Delta$}
        \If{$(u,x) \in \Samp_e \wedge (u,w) \in \Samp_e  \wedge (u,z) \in \Samp_e $}
      		\State $p \leftarrow  \frac{M_C}{t_C} \cdot (\frac{M_e}{t})^3  \cdot (\frac{M_e}{t^*})^5$
      		\State $\sigma \leftarrow \varkappa + p^{-1}/2$
     \EndIf   
      \EndFor
    \EndFunction
    \Statex
    \Function{Update4Cliques}{$(u,v), t $}
      \State $\mathcal{N}^\Samp_{u,v} \leftarrow \mathcal{N}^\Samp_u \cap \mathcal{N}^\Samp_v$ 
      \For{ {\bf each} pair $(w,z)$ from $\mathcal{N}^\Samp_{u,v}\times\mathcal{N}^\Samp_{u,v}$}
       	 \If{$(w,z) \in \Samp_e$}
			\State $t_C \leftarrow t_C + 1$
      	\If{\textsc{Sample4Clique}$(u,v,w,z,t)$}
      		\State $\Samp_C \leftarrow \Samp_C\cup \{u,v,w,z,t\}$
      	\EndIf
     \EndIf   
	
      \EndFor
    \EndFunction
    \Statex
    \Function{\textsc{Sample4Clique}}{$u,v,w,z,t$} 
         \If {$t_C\leq  M_C$}
        \State \textbf{return} True
      \ElsIf{\textsc{FlipBiasedCoin}$(\frac{M_C}{t_C}) = $ heads}
        \State $(u_1,v_1,w_1, z_1, t_1) \leftarrow$ random 4-clique from $\Samp_C$
        \State $\Samp_C\leftarrow \Samp_C\setminus \{(u_1,v_1,w_1,z_1,t_1)\}$
        \State \textbf{return} True
      \EndIf
      \State \textbf{return} False
    \EndFunction
    \Statex
    \Function{\textsc{SampleEdge}}{$(u,v),t$} 
      \If {$t\leq  M$}
        \State \textbf{return} True
      \ElsIf{\textsc{FlipBiasedCoin}$(\frac{M}{t}) = $ heads}
        \State $((u',v'),t') \leftarrow$ random edge from $\Samp$
        \State $\Samp\leftarrow \Samp\setminus \{((u',v'),t')\}$
        \State \textbf{return} True
      \EndIf
      \State \textbf{return} False
    \EndFunction
  \end{algorithmic}
\end{algorithm}

\begin{algorithm}[t]
\caption{\algosinglefive{} - Single Reservoir Sampling for 5-cliques counting}
  \label{alg:fivest}
   \begin{algorithmic}[0]
	\Statex{\textbf{Input:} Edge stream $\Sigma$, integer $M\ge6$}
	\Statex{\textbf{Output:} Estimation of the number of 5-cliques $\varkappa$}
    \State $\Samp_e \leftarrow\emptyset$, $t\leftarrow 0$, $\varkappa\leftarrow 0$ 
    \For{ {\bf each} element $(u,v)$ from $\Sigma$}
    \State $t\leftarrow t +1$
    \State \textsc{Update5Cliques}$(u,v)$
    \If{\textsc{SampleEdge}$((u,v), t )$} 
      \State $\Samp \leftarrow \Samp\cup \{(u,v)\}$
    \EndIf
    \EndFor
    \Statex
    \Function{Update5Cliques}{$u,v$}
      \State $\mathcal{N}^\Samp_{u,v} \leftarrow \mathcal{N}^\Samp_u \cap \mathcal{N}^\Samp_v$
      \For{ {\bf each} element $(x,w,z)$ from $\mathcal{N}^\Samp_{u,v} \times \mathcal{N}^\Samp_{u,v} \times \mathcal{N}^\Samp_{u,v}$}
       \If{$\{(x,w),(x,z),(w,z)\}\subseteq\Samp_e$} 
       		\If{$t\leq M+1$}
       			\State{$p\leftarrow 1$}
       		\Else
      			\State $p \leftarrow  \prod_{i=0}^{9}\frac{M-i}{t-i-1} $
      		\EndIf
      	\State $\varkappa \leftarrow \varkappa + p^{-1}$
    \EndIf
	\EndFor
   \EndFunction
\end{algorithmic}
\end{algorithm}

\begin{lemma}\label{lem:singlecountfive}
	Let $\lambda\in \clisett{5}{t}$ with $\lambda = \{e_1,\ldots,e_{10}\}$. Assume, without loss of generality, that the edge $e_i$ is observed at $t_i$ (not necessarily consecutively) and that $t_{10}>\max\{t_i,1\leq i\leq 9\}$. $\lambda$ is observed by \algosinglefive{} at time $t_{10}$ with probability:
	\begin{equation}
		p_\lambda = 
		\begin{cases}
			0&\ if |M|<9,\\
			1&\ if\ t_{10}\leq M+1,\\
			\prod_{i=0}^{9}\frac{M-i}{t-i-1}&\ if\ t_{10}>M+1.	
 		\end{cases}
	\end{equation}
\end{lemma}

\begin{theorem}\label{thm:singlefiveunbiased}
  Let $\varkappa^{(t)}$ the estimated number of 5-cliques in $G^{(t)}$ computed by \algosinglefive{} using memory of size $M>9$. $\varkappa^{(t)}=|\clisett{5}{t}|$ if $t\le M+1$ and
  $\mathbb{E}\left[\varkappa^{(t)}\right]$ $=|\clisett{5}{t}|$ if $t> M+1$.
\end{theorem}

The proof for Lemma~\ref{lem:singlecountfive} (resp., Theorem~\ref{thm:singlefiveunbiased}), closely follows the steps of the proof of Lemma~\ref{lem:singlecount} (resp., Theorem~\ref{thm:singleunbiased})).

\end{document}